\documentclass[aps,superscriptaddress,twocolumn,nofootinbib,showpacs,notitlepage]{revtex4-1}
\usepackage{amsfonts}
\usepackage{amsmath}
\usepackage{appendix}
\usepackage{amsthm}
\usepackage{braket}
\usepackage{graphics}
\usepackage[hidelinks]{hyperref}
\usepackage{physics}
\usepackage{color}
\usepackage[dvipsnames]{xcolor}
\usepackage{amssymb}
\usepackage{dsfont}
\usepackage{verbatim}
\usepackage{amsmath,amscd}
\usepackage[all,cmtip]{xy}
\usepackage{ragged2e}
\usepackage{multirow}
\usepackage{bbm} 
\usepackage[capitalize]{cleveref} 
\usepackage{graphicx}
\usepackage{float}
\usepackage[caption=false,justification=centerlast]{subfig}
\usepackage[mathscr]{euscript}
\usepackage[normalem]{ulem}
\graphicspath{{./}}

\usepackage{tikz}
\usetikzlibrary{shapes.geometric, arrows}
\tikzstyle{decision} = [diamond, draw, fill=blue!20, 
    text width=4.5em, text badly centered, node distance=3cm, inner sep=0pt]
\tikzstyle{block} = [rectangle, draw, fill=blue!20, 
    text width=5em, text centered, rounded corners, minimum height=4em]
\tikzstyle{line} = [draw, -latex']

\newtheorem{theorem}{Theorem}[section]
\newtheorem{lemma}[theorem]{Lemma}

\newtheorem{definition}[theorem]{Definition}
\newcommand{\Id}{\mathbbm{1}} 
\newcommand{\supp}{\operatorname{Supp}}
\newcommand{\defi}{:=} 

\makeatletter
\newcommand{\crefnames}[3]{%
  \@for\next:=#1\do{%
    \expandafter\crefname\expandafter{\next}{#2}{#3}%
  }%
}
\makeatother
\crefnames{section}{Section}{Sections}
\definecolor{darkblue}{RGB}{0,0,127}

\usepackage[utf8]{inputenc}
\usepackage[T1]{fontenc}

\makeatletter
\renewcommand\@make@capt@title[2]{%
  \@ifx@empty\float@link{\@firstofone}{\expandafter\href\expandafter{\float@link}}%
   {\textbf{#1}}\@caption@fignum@sep#2\quad
}%
\makeatother

\begin{document}

\title{Fault-tolerant logical gates in holographic stabilizer codes are severely restricted} %quantum error-correcting codes}

\author{Sam Cree}
\email{scree@stanford.edu}
\affiliation{Stanford Institute for Theoretical Physics, Stanford University, Stanford, CA 94305, USA}
\author{Kfir Dolev}
\affiliation{Stanford Institute for Theoretical Physics, Stanford University, Stanford, CA 94305, USA}
\author{Vladimir Calvera}
\affiliation{Department of Physics, Stanford University, Stanford, CA 94305, USA}
\author{Dominic J. Williamson}
\affiliation{Stanford Institute for Theoretical Physics, Stanford University, Stanford, CA 94305, USA}

\begin{abstract}
	We evaluate the usefulness of holographic stabilizer codes for practical purposes by studying their allowed sets of fault-tolerantly implementable gates.
	We treat them as subsystem codes and show that the set of transversally implementable logical operations is contained in the Clifford group for sufficiently localized logical subsystems.
	As well as proving this concretely for several specific codes, we argue that this restriction naturally arises in any stabilizer subsystem code that comes close to capturing certain properties of holography.
	We extend these results to approximate encodings, locality-preserving gates, certain codes whose logical algebras have non-trivial centers, and discuss cases where restrictions can be made to other levels of the Clifford hierarchy.
	A few auxiliary results may also be of interest, including a general definition of entanglement wedge map for any subsystem code, and a thorough classification of different correctability properties for regions in a subsystem code.
\end{abstract}

\maketitle

\tableofcontents

\section{Introduction}

The anti-de Sitter/conformal field theory (AdS/CFT) correspondence~\cite{Maldacena1997,Witten1998} is a particularly fruitful instance of the holographic principle~\cite{Hooft1993,Susskind1995} from high-energy theory -- it refers to a duality between a theory of quantum gravity in $D+1$ dimensions and a conformally-invariant quantum field theory on the $D$-dimensional boundary of the original space.
It was observed in Ref.~\cite{almheiriBulkLocalityQuantum2015} that this duality has many properties in common with quantum error-correcting codes, which motivated the development of several families of error-correcting codes built out of tensor networks~\cite{pastawskiHolographicQuantumErrorcorrecting2015,haydenHolographicDualityRandom2016,donnellyLivingEdgeToy2017,caoApproximateBaconShorCode2020,farrellyParallelDecodingMultiple2020,harrisCalderbankSteaneShorHolographicQuantum2018,harrisMaximumLikelihoodDecoder2020,Cao2021,Jahn2021} as toy models for the holographic correspondence.
These \textit{holographic codes} encode information into a series of degrees of freedom which are increasingly non-local, with the degree of non-locality manifesting in the logical system as an additional emergent geometric dimension.
The logical system is associated with this higher-dimensional ``bulk'' geometry, and the physical system with the lower-dimensional ``boundary''.

Quantum error correction is a remarkable idea that unlocks the possibility of feasible quantum computation~\cite{Shor1995,Aharonov2008,shor1996fault}. 
Many of the leading proposals for implementing quantum error correction are based on a concept borrowed from theoretical condensed matter physics, namely topological order~\cite{wegner1971duality,PhysRevLett.50.1395,doi:10.1142/S0217979290000139}. 
These \textit{topological codes}~\cite{qdouble,Bravyi1998,Freedman2001} encode information into the (necessarily non-local) degrees of freedom that are sensitive to the topology of the surface on which the code lives. Thus they protect against local, topologically trivial, operations. 
The success of topological codes motivates us to assess the usefulness of these new holographic codes for practical quantum computation.

There are a number of reasons to suspect that holographic codes may be of practical use for quantum computing.
Holographic codes can admit erasure thresholds comparable to that of the widely-studied surface code~\cite{pastawskiHolographicQuantumErrorcorrecting2015,harrisCalderbankSteaneShorHolographicQuantum2018}, and likewise for their threshold against Pauli errors~\cite{harrisMaximumLikelihoodDecoder2020}. 
Their holographic structure also naturally leads to an organization of encoded qubits into a hierarchy of levels of protection from errors, which could be useful for applications which call for many qubits with varying levels of protection.
In particular, this is reminiscent of many schemes for magic state distillation~\cite{Knill2004,Bravyi2005} -- and indeed, the concatenated codes~\cite{Knill1996} utilized for magic state distillation share a similar hierarchical structure to holographic codes. 
The layered structure of holographic codes is also reminiscent of memory architectures in classical computers, where it is useful to have different levels of short- and long-term memory. 
Although these codes have some notable drawbacks, in particular holographic stabilizer codes require non-local stabilizer generators, other codes such as concatenated codes suffer similar drawbacks and have still proven to be useful.
Conversely, the stringent requirement of non-local stabilizer generators allows holographic codes to protect many more qubits than a topological code and in fact attain a finite nonzero encoding rate, which is typically not possible for topological codes.
Nonetheless, many open questions remain about the usefulness of holographic codes for fault-tolerant quantum computing.

Beyond the robust storage of quantum information, a desirable feature of a quantum codes is the ability to perform a large set (ideally a universal set) of logical operations while maintaining the protection provided by the code. 
Thus it is useful to study the set of \textit{fault-tolerant} logic gates implementable in a given code, i.e.\ those that confine the spread of a physically relevant set of correctable errors sufficiently well that they remain correctable. 
The simplest examples are \textit{transversal} gates, which do not couple the physical subsystems of the code and hence do not spread errors at all. 
The desirable error correction properties of these gates come at a cost, which is that they are restricted to implement a non-universal (and in fact, finite) group of logical operations on a code with non-trivial distance \cite{eastinRestrictionsTransversalEncoded2009}.

Topological codes possess additional geometric structure, which results in a wider class of fault-tolerant gates that subsumes the set of transversal gates: those implementable by \textit{locality-preserving} operations.
Locality-preserving operations are those which, when acting on a local operator via conjugation, expand the region on which it acts by a bounded amount.
These preserve the correctability of errors confined to sufficiently small regions. 
However, like transversal gates, locality-preserving gates in topological codes are also restricted to form a finite group\footnote{
This does not include the more general class of gates and codes with boundaries considered in Refs.~\cite{Zhu2017,Lavasani2019}.}~\cite{bravyiClassificationTopologicallyProtected2013,beverlandProtectedGatesTopological2016}.  
Specifically, it has been shown that locality-preserving gates in topological \textit{stabilizer subsystem codes}~\cite{bravyiClassificationTopologicallyProtected2013,pastawskiFaulttolerant2015,Webster2018} in $D$ spatial dimensions can implement only gates from $\mathcal{C}_D$, the $D$-th level of the Clifford hierarchy, which is a series of increasing sets of gates that includes the Pauli group ($\mathcal{C}_1$) and the Clifford group ($\mathcal{C}_2$). 
These sets obey $\mathcal{C}_1 \subset \mathcal{C}_2 \subset \dots $, and the bound is known to be saturated; i.e.\ for all $D>1$, one can find a $D$-dimensional topological stabilizer (subsystem) code that can transversally implement a gate from $\mathcal{C}_D \backslash \mathcal{C}_{D-1}$, for example the color code~\cite{bombin2015gauge,https://journals.aps.org/prl/abstract/10.1103/PhysRevLett.98.160502}.

Just as for topological codes, the geometric structure of holographic codes endows locality-preserving gates with fault tolerance. 
This raises a natural question: what restrictions does the holographic structure of these codes impose on locality-preserving gates?
A natural guess for holographic \textit{stabilizer} codes in $D$ spatial dimensions might be that one can access $\mathcal{C}_{D+1}$, the $(D+1)$-th level of the Clifford hierarchy, due to their locality in the extra holographic dimension and the analogy to topological codes. 
Here we show that the restrictions are in fact much more severe than this; namely that locality-preserving gates for such holographic codes are dimension-independently restricted to the Clifford group, $\mathcal{C}_2$.

\subsection{Summary of results \& main ideas}

In this work, we show that for sufficiently locality-preserving physical gates, the restricted logical action on a sufficiently small region of the bulk cannot implement a non-Clifford gate for any stabilizer code with a holographic structure that captures certain properties of AdS/CFT in any spatial dimension.
Deeper into the bulk, this restriction applies even to unitaries that expand the support of local operators by increasingly large amounts.
We describe a number of generalizations of our result, including an extension to approximate stabilizer encoding maps and to the codes studied in Ref.~\cite{donnellyLivingEdgeToy2017} with logical algebras whose centers are non-trivial; in all these cases we show that it remains impossible to implement a logical non-Clifford gate.
We also show that as our assumptions are relaxed then even if non-Clifford gates are allowed, the logical gate is still restricted to some higher level of the Clifford hierarchy determined by the size of the logical bulk region and the degree of spreading by the locality-preserving unitary.

We now explain the relatively simple idea that underlies most of our results. 
First, we recall the following implication of a lemma due to Pastawski and Yoshida stated informally: any stabilizer subsystem code whose physical qudits can be partitioned into three correctable regions can only support transversal Clifford gates (see below for the technical statement, \cref{lem:PY}). 
For any given instance of a holographic code and choice of small bulk subregion, it is typically easy to find a partitioning of the physical system into three correctable regions, meaning this lemma can be applied.
The main contribution of this work is to argue from the physical principles of holography when and why it is possible to find such a tripartition.
A key error-correcting property of holographic codes is a relationship between the correctability of two complementary regions, known as complementary recovery.
We argue that so long as regions that \textit{almost} satisfy complementary recovery are not \emph{too} rare, then the partitioning into three correctable regions is possible -- see Fig.~\ref{fig:MainIdea}. 
We go on to make use of the geometric structure and locality inherent to the physical qubits at the boundary of a holographic code to extend our results to cover locality-preserving unitaries. 
Informally, we show the following.
\begin{theorem}[Informal Statement of \cref{thm:1d}]
    If there exists an appropriate region that almost satisfies complementary recovery, and the protected bulk subsystem is sufficiently small, then any sufficiently locality-preserving unitary with well-defined logical action implements an element of the Clifford group. 
\end{theorem}

We remark that there already exists a bound~\cite{pastawskiFaulttolerant2015} establishing that transversal unitaries acting on logical subsystems that have sufficiently large distances, such as those associated with qubits deep in the bulk, are already confined to a particular level of the Clifford hierarchy.
Our work can be seen as a strengthening of that result in the specific context of holographic codes.
We apply it to non-transversal but locality-preserving gates, as well as logical qubits associated to more shallow regions of the bulk whose distances are not sufficiently large for the previously established arguments to apply.

Our results share some similarities with recent work ruling out global symmetries in AdS/CFT~\cite{harlowSymmetriesQuantumField2019,Harlow2018}. 
The results proved there imply that holographic codes satisfying a set of assumptions motivated by the AdS/CFT correspondence do not support any non-trivial locality preserving gates that act uniformly on the bulk encoded qubits. 
Our results differ on one hand as we require weaker assumptions to be satisfied by the holographic codes in question, and we consider a broader class of locality-preserving gates that may act on the bulk in a non-uniform fashion.
As a consequence, our results are in a sense weaker where they overlap, only ruling out non-Clifford global symmetries in the bulk rather than all global symmetries.
Our broader assumptions necessarily require this weakening of the results, as our notion of holographic codes includes the well-known HaPPY code~\cite{pastawskiFaulttolerant2015}, which is known to violate the stronger ``no global symmetries'' result by admitting non-trivial Clifford global symmetries%
\footnote{%
Despite the no-go result in Ref.~\cite{harlowSymmetriesQuantumField2019}, such global symmetries are in fact possible when domain walls of the boundary symmetry cannot be locally mapped back into the code subspace. 
This does not contradict Ref.~\cite{harlowSymmetriesQuantumField2019} as it violates the assumptions made there, we plan to discuss this further in a forthcoming work~\cite{unpublishedGauge}. 
}.
On the other hand, we also assume additional structure by requiring the codes to be stabilizer codes,
which allows us to make stronger claims about bulk logical operations that are not global symmetries.

As far as practical applications are concerned, unfortunately our main results point out an obstacle for any potential application of holographic stabilizer codes for magic state distillation. 
Specifically, we show that such codes essentially do not support transversal logical gates outside the Clifford group as required for magic state distillation.

\begin{figure}[t]
	\centering
	\includegraphics[width=1.0\linewidth]{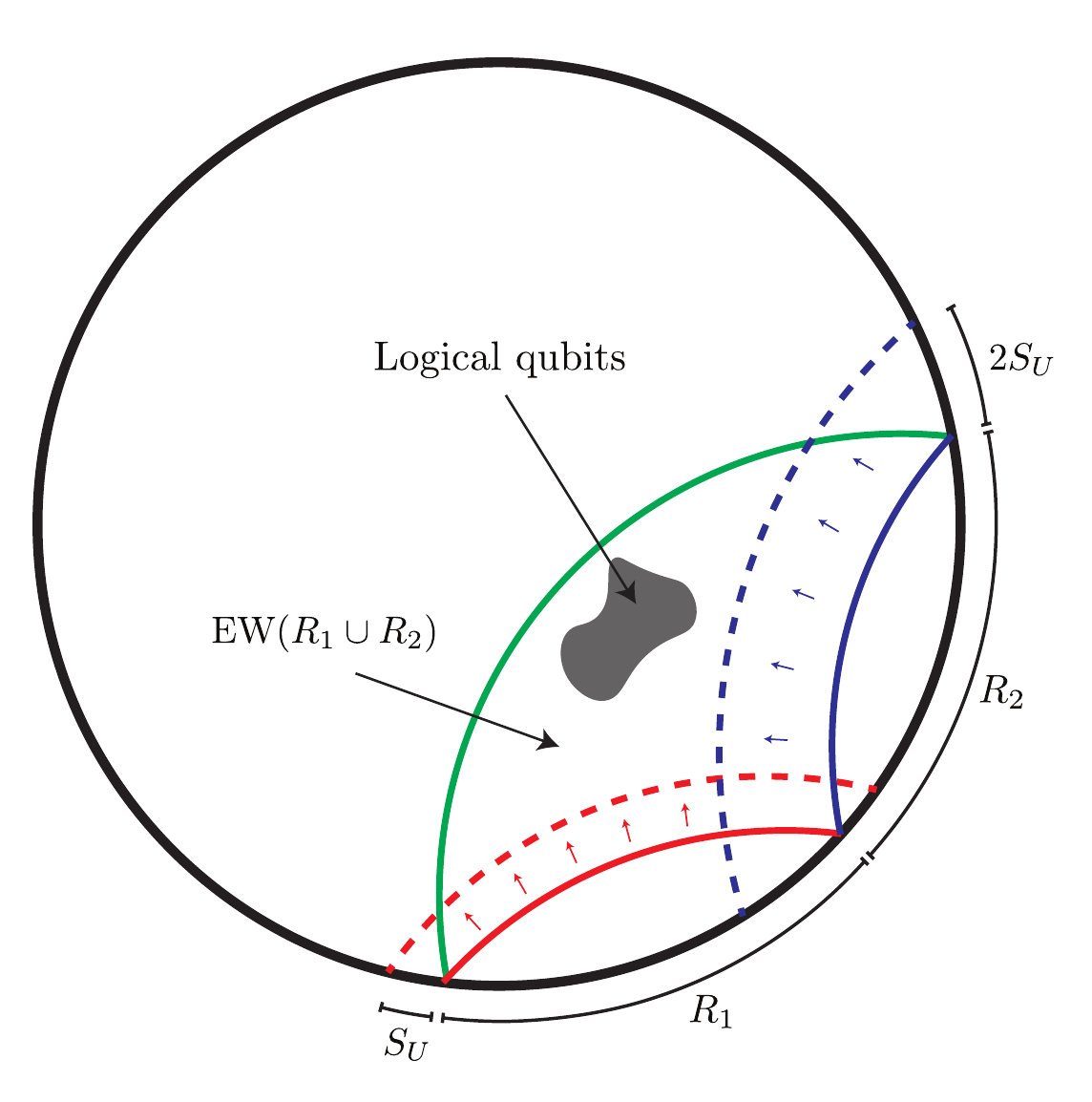}
	\caption{A partition of the physical qubits on the boundary into three regions $R_1,R_2,(R_1 \cup R_2)^c,$ that are all correctable with respect to some logical qubits associated to a bulk region due to the holographic structure of the code. 
    In this setting locality preserving logic gates $U$ that only spread local operators by an amount $S_U$ can only implement Clifford gates on the logical qubits.} 
\label{fig:MainIdea}
\end{figure}

\subsection{Outline}

The manuscript is laid out as follows: We first introduce background on subsystem codes in \cref{sect:background} and prove some of their properties that are useful for our purposes. 
We also introduce a general definition of entanglement wedge, termed the \emph{maximal entanglement wedge}, which is defined for an arbitrary subsystem code without any assumptions of additional structure such as a particular geometry.
This contributes to a growing set of such results~\cite{harlowRyuTakayanagiFormulaQuantum2017,pastawskiCodePropertiesHolographic2017,kamal2019ryu} which identify holographic features that emerge from surprisingly simple quantum error-correction properties.
We then describe some of the properties that we associate with holographic codes, such as complementary recovery and a geometric structure.
We also review the stabilizer formalism, to emphasize some crucial subtleties that allow for a natural definition of the Clifford group on the logical space.
Finally, we review background on fault-tolerant implementation of logical gates; in particular a crucial lemma from Ref.~\cite{pastawskiFaulttolerant2015} (the Pastawski-Yoshida lemma) that shows that the existence of regions with certain error-correction properties implies a restriction on the logical gates implementable by locality-preserving operations.

In \cref{sect:happy}, we consider a prototypical family of holographic codes constructed from perfect tensors; known as HaPPY codes after the authors of Ref.~\cite{pastawskiHolographicQuantumErrorcorrecting2015}.
We show that any code made from copies of a single perfect tensor must support appropriate regions to apply the Pastawski-Yoshida lemma; and argue that this property extends to any HaPPY code built from more general combinations of perfect and planar-perfect tensors.

Since the introduction of HaPPY, many other families of holographic stabilizer codes have been introduced which extend beyond the assumptions made in Ref.~\cite{pastawskiHolographicQuantumErrorcorrecting2015}. 
Thus, we argue in \cref{sect:complementary} on more general grounds that any holographic code that even almost captures the code properties of AdS/CFT (\emph{complementary recovery} in particular) must also have their locality-preserving gates restricted to only implement Clifford gates.
We also extend this argument to approximate codes such as that in Ref.~\cite{caoApproximateBaconShorCode2020},  and to codes whose logical algebras may have non-trivial centers~\cite{donnellyLivingEdgeToy2017}.
Finally, in \cref{sect:levels} we show variants of this argument for other levels of the Clifford hierarchy.
For example, even if a bulk region is sufficiently large that the arguments of \cref{sect:complementary} do not apply, we can still establish a weaker restriction to a level $\mathcal{C}_n$ for some $n>2$.
Interestingly, in other cases we are able to show an even stronger restriction than our main result. 
In these cases, a codespace-preserving operator that acts only on the logical subsystem and trivially elsewhere can not even implement a Clifford gate unless it is a Pauli; i.e.\ it is restricted to $\mathcal{C}_1$.

\section{Preliminaries} \label{sect:background}

In this section, we introduce various properties and additional structures on error-correcting codes that are necessary to formulate and prove the main results of our work in later sections.
We review subsystem and stabilizer codes, introduce the error-correction and geometric properties of holographic codes, and review some useful lemmas about locality-preserving gates.

In holographic codes, the bulk space (which we denote $B$) is associated with the encoded space of the code, and the boundary space (denoted $\partial B$) with the physical degrees of freedom of the code.
Different encoded/bulk qubits have different levels of protection against error.
Those near the boundary are susceptible to small errors occurring nearby on the boundary, whereas those deep in the bulk are well protected against all small errors.
For this reason, it is useful to characterize the error-correction properties of such a code with respect to a choice of a subset of the encoded qubits that are intended to be protected while the remainder need not be. 
This can be accomplished with the subsystem code formalism, which we introduce here.

\subsection{Subsystem codes}\label{subsec:subsystem-codes}
For a general error-correcting code, the physical Hilbert space $\mathcal{H}_p$ decomposes into $\mathcal{H}_p= \mathcal{H}_c \oplus \bar{\mathcal{H}}_c$, with $\mathcal{H}_c$ the code subspace and $\bar{\mathcal{H}}_c$ its complement.
For a standard code, the code subspace is taken to be isomorphic to a logical system $\mathcal{H}_L$ which represents the important encoded quantum information.
In the more general framework of \emph{subsystem} codes, $\mathcal{H}_c$ is isomorphic to a tensor product $\mathcal{H}_L \otimes \mathcal{H}_J$, with the logical subsystem $\mathcal{H}_L$ storing the ``important information'', and the \emph{junk} subsystem (sometimes known as the \emph{gauge} subsystem)\footnote{Operations acting on the gauge subsystem that act trivially on the logical subsystem are unimportant, as we do not care about the state of the gauge system -- in this way, they are similar to gauge transformations in gauge theory, hence the name.} $\mathcal{H}_J$ storing ``unimportant information''. 
This division is particularly useful when different subsystems of the code space experience different levels of protection against error, as one can evaluate the error-correction properties of $\mathcal{H}_L$ independently of $\mathcal{H}_J$, which may be highly susceptible to errors.
We sometimes refer to the combined system $\mathcal{H}_L \otimes \mathcal{H}_J$ as the \emph{encoded space}, $\mathcal{H}_e = \mathcal{H}_L \otimes \mathcal{H}_J \cong \mathcal{H}_c$.

The encoding isometry $V: \mathcal{H}_L \otimes \mathcal{H}_J \to \mathcal{H}_p$ maps the combination of these two subsystems into the physical Hilbert space.
Its image is the code subspace, $\mathcal{H}_c$, and we denote the projector onto this subspace as $\Pi = VV^{\dagger} \in \mathcal{B}(\mathcal{H})$, where $\mathcal{B}(\mathcal{H})$ denotes the set of operators acting on the Hilbert space $\mathcal{H}$.
An operator $A\in \mathcal{B}(\mathcal{H})$ is said to be \emph{codespace-preserving} (CSP) if it commutes with the code projector, $[A,\Pi] = 0$.
As we only aim to protect the logical subsystem from error, a codespace-preserving operator that entangles the two encoded subsystems is undesirable, as it can contaminate the important information with junk information that has not been well-protected from error.
Thus we focus on \emph{dressed-CSP} operators, which are CSP operators $A$ that act as a product operator on the two subsystems, i.e.\ $V^{\dagger}\! A V = A_L \otimes A_J \in \mathcal{B}\left( \mathcal{H}_L \otimes \mathcal{H}_J \right)$.
We say that such an operator $A$ \emph{implements} the logical operation $A_L\otimes A_J$, which is a \emph{dressed-logical} operator.
When the operator acts trivially on the junk subsystem, i.e.\ $A_J = \Id_J$, we say that $A$ is a \emph{bare-CSP} operator, and its implemented operator $A_L \otimes \Id_J$ is a \emph{bare-logical} operator.
We say that a pair of bare- or dressed-CSP operators $A$ and $B$ are \emph{logically equivalent} if, up to scalar multiplication, they act the same way on the logical subsystem,  i.e.\ $A_L \propto B_L$.

Let us now discuss error-correction properties of subsystem codes.
We present three equivalent conditions for a particular region $R$ (that is, a collection of physical subsystems) to be \emph{correctable}\footnote{
One can generalize this definition by considering correctability with respect to any particular noise channel, rather than just erasure channels as we do here. However we are primarily interested in the implications of the geometric local structure of holographic codes, which is most readily apparent when studying correctability with respect to erasure of particular regions.} -- i.e.\  when it is possible to correct for the erasure of $R$.
The conditions are expressed in terms of operators 
that are supported in $R$, by which we mean those elements of $\mathcal{B}(\mathcal{H})$ that take the form $A_R \otimes \Id_{R^c} $.

\begin{lemma}
	For a region $R$ of a physical Hilbert space of a code with projector $\Pi$, the following properties are equivalent \cite{pastawskiFaulttolerant2015}:
\begin{itemize}
	\item There exists a recovery channel $\mathcal{R}$ such that for every bare-CSP operator $A$, $\Pi (\mathcal{R} \circ \mathcal{N}_R)^{\dagger}(A) \Pi = \Pi A \Pi$, with $\mathcal{N}_R(\cdot ) \defi  \tr_R(\cdot )$ the channel that completely erases subsystem $R$.
	\item For any operator $ B_R $ supported in $R$, and any bare-CSP operator $A$, $[A,\Pi B_R \Pi] = 0$.
	\item For any logical operator $A_L$, there exists a bare-CSP operator $A$ supported in $R^c$ such that $A$ implements $A_L \otimes \Id_J$.
	\end{itemize}
	\label{lem:correctable}
\end{lemma}
\begin{definition}
	A region $R$ of the physical Hilbert space is said to be \textbf{\emph{correctable}} if the above conditions hold.
\end{definition}

We remark that each of these conditions refers to \emph{bare-CSP} operators, which is in turn dependent on the choice of logical subsystem $\mathcal{H}_L$. 
Some regions may be correctable for some choices of $\mathcal{H}_L$ and not others.
In the case that the junk subsystem is trivial, i.e.\ $\mathcal{H}_L \equiv \mathcal{H}_c$, then the equivalence of the first two points in \cref{lem:correctable} corresponds to the familiar Knill-Laflamme error-correction conditions \cite{Knill2004}, and the equivalence of the third point is the so-called \emph{cleaning lemma} for stabilizer codes \cite{bravyi2009no}.

In standard codes, there is another equivalent property to the three presented in \cref{lem:correctable}, which is that any CSP operator $A$ supported in $R$ implements the logical identity up to a scalar, $V^{\dagger}AV \propto \Id$. 
A similar property holds for correctable regions in subsystem codes, although it is not known to be a sufficient criterion for correctability in this context.
\begin{lemma}
	Suppose a region $R$ is \emph{correctable}. 
	Then any \emph{dressed-CSP} operator supported on $R$ implements a logical operation of the form $\Id_L \otimes A_J$, for some $A_J \in \mathcal{B}(\mathcal{H}_J)$.
	\label{lem:dressedtrivial}
\end{lemma}

\begin{proof}
Suppose $R$ is correctable. 
Then take any dressed-CSP operator $A$ supported on $R$ with $V^{\dagger} A V = A_L \otimes A_J $.
By assumption, for any $B_L \in \mathcal{B}(\mathcal{H}_L)$, one can find $B$ supported on $R^c$ such that $V^{\dagger} B V = B_L \otimes \Id_J $.
Since $A$ and $B$ are supported on complementary regions, we have $[A,B]=0$, which again implies $[A_L, B_L ] = 0$.
Since this is true for any $B_L \in \mathcal{B}(\mathcal{H}_L)$, $A_L$ must be proportional to the identity, $A_L = c\Id_L$.
\end{proof}

\subsection{Entanglement wedge in subsystem codes} \label{subsec:holocodes-as-subsystem-codes}
The encoded space of a subsystem code factorizes into a tensor product of a logical subsystem and a junk subsystem.
More generally, it typically has a richer tensor product structure in the form of a decomposition into local encoded subsystems 
\begin{align}
    \mathcal{H}_e = \bigotimes_{i=1}^{n_L} \mathcal{H}_i, \label{eq:TPS}
\end{align}
e.g.\ individual qudits, sites of a lattice, etc. 
Then each subset $S \subset \left\{ 1,\dots ,n_L \right\}$ corresponds to a choice of logical subsystem $\mathcal{H}_L = \bigotimes_{ i\in S} \mathcal{H}_i$, with junk subsystem $\mathcal{H}_J = \bigotimes_{i\notin S} \mathcal{H}_i$.
We then say that a region is correctable with respect to $S$ if it is correctable with respect to the corresponding choice of logical and junk subsystems.

The correctability of a physical region is preserved under union of logical subsystems, i.e.\ if a region $R$ is correctable with respect to $S_1$ and $S_2$, it is also\footnote{
This follows from the fact that any operator $A$ supported on the Hilbert space associated with $S_1 \cup S_2$ can be written as a linear combination of product operators on the spaces associated with $S_1$ and $S_2 \backslash S_1$, $A = \sum_{i}^{} A_{1}^i \otimes A_{2\backslash 1}^i$, with an implicit identity acting on systems not labelled.  By correctability of $R$ w.r.t.\ $S_{1,2}$, each operator $A_1^i $ and $A_{2\backslash 1}^i$ can be implemented by a CSP operator supported on $R^c$.  These operators can be multiplied and linearly combined appropriately to implement $A$ with a CSP operator on $R^c$, and thus $R$ is correctable w.r.t.\ $S_1 \cup S_2$.} 
correctable with respect to $S_1 \cup S_2$.
This implies the existence of a unique ``maximal'' logical subsystem given by the union of all local encoded subsystems with respect to which $R$ is correctable. 
We define this to be the \emph{maximal entanglement wedge} of $R^c$ in the subsystem code.
\begin{definition}
	The \textbf{\emph{maximal entanglement wedge}}, denoted $\mathcal{E}[R]$, is defined to be the union of all local encoded subsystems with respect to which $R^c$ is correctable. 
	\label{def:ew}
\end{definition}
Let us justify the terminology used here by comparing to the entanglement wedge as defined in holography.
Any operator acting on local encoded subsystems in the entanglement wedge of $R$ can be reconstructed in $R$, just as in holography.
Furthermore, any alternative definition of an entanglement wedge with this property (for example, the greedy entanglement wedge map of HaPPY), must be included within the maximal entanglement wedge as defined above%
\footnote{We conjecture that our definition in fact coincides with that of the greedy entanglement wedge from HaPPY -- essentially, perfect tensors are already so highly constrained that the entire algebra of operators on a subsystem can only be correctable if this is already guaranteed by the properties of perfect tensors.
}.

We remark that this definition of maximal entanglement wedge makes no additional assumptions beyond a particular choice of decomposition into local encoded subsystems.
Thus the existence of an entanglement wedge map is a general feature of any subsystem code.

Relevant to our purposes of assessing the usefulness of holographic codes for fault-tolerant quantum computation, the maximal entanglement wedge suggests a correlation between the usefulness of a subsystem code for practical error-correction against erasure of $R^c$, and its quality as a model of holography -- namely, both of these characteristics are improved when the entanglement wedge of $R$ is larger.
The former should be clear; a larger entanglement wedge means that $R^c$ is correctable with respect to more subsystems, and so more encoded information can be protected from erasures.

As for the latter, we remark that there is a limit to how large entanglement wedges can be -- a limit which, it turns out, AdS/CFT saturates.
A given local encoded subsystem cannot belong to both the maximal entanglement wedge of a boundary region $R$ and its complement $R^c$, as this would imply that its bare logical operators can be reconstructed in both $R$ and $R^c$, which would mean that all its bare logical operators must commute; a contradiction.
Thus we have in general,
\begin{align}
    \mathcal{E}[R] \subseteq \mathcal{E}[R^c]^c.
    \label{eq:ewlimit}
\end{align}
If this limit is saturated for some $R$, that is ${\mathcal{E}[R] = \mathcal{E}[R^c]^c}$, then we say that $R$ obeys \emph{complementary recovery}.
\begin{definition}
	If all local encoded subsystems belong to either $\mathcal{E}[R]$ or $\mathcal{E}[R^c]$, then $R$ is said to obey \textbf{\emph{complementary recovery}}.
	\label{def:cr}
\end{definition}
In full AdS/CFT, the entanglement wedge map is expected to satisfy this property for any region $R$ -- in this case we say that the code itself obeys complementary recovery.
In fact, it is known that a subsystem code with complementary recovery automatically exhibits the Ryu-Takayanagi formula from holography~\cite{harlowRyuTakayanagiFormulaQuantum2017}. 
Thus, codes with larger entanglement wedges -- closer to the limit set by \cref{eq:ewlimit} -- are both better at protecting from errors%
\footnote{More precisely we mean that for a local encoded subsystem with a fixed price, a code that satisfies complementary recovery leads to the maximum possible distance for that subsystem.}%
, \emph{and} more closely resemble holography. 

Following this line of reasoning, several previous works have defined holographic codes to be those that satisfy complementary recovery.
As we show later in \cref{thm:complementary}, our results are  particularly straightforward to state for these codes.
However, this class of codes, which we refer to as \emph{ideal holographic codes}, is too restrictive for our purposes. 
In particular, the most notable (and original) example of a holographic stabilizer code, the HaPPY code, is excluded from this class. 
It is relatively easy in HaPPY to construct a region with many disconnected pieces, such that the majority of the bulk qubits lie in neither its entanglement wedge nor that of its complement -- an example in \cref{fig:empty} is shown where both entanglement wedges are completely empty.
One can also construct connected regions in HaPPY codes which do not obey complementary recovery, for example see \cref{fig:almostcomplementary}; however these typically have only a small fraction of the bulk that lies in neither entanglement wedge.
Furthermore, many regions do not have any such residual bulk region, as in \cref{fig:complementary}.
Other holographic codes with promising attributes for practical application similarly fail to meet the stringent criterion of an ideal holographic code~\cite{harrisCalderbankSteaneShorHolographicQuantum2018,harrisMaximumLikelihoodDecoder2020}.
Thus for our purposes, a holographic code is one which (\textit{almost}) saturates
\cref{eq:ewlimit} for many regions $R$; the more common this property is among the code's regions, the closer to an ideal holographic code it becomes.
For example, connected regions in HaPPY typically fail to satisfy complementary recovery by a small fraction of the total number of bulk sites, tending to zero in the worst case\footnote{The size of this residual region for a connected boundary component was in fact argued to be constant in Ref.~\cite{pastawskiHolographicQuantumErrorcorrecting2015}.} as $\log |R|/|B|\rightarrow 0$ as the number of sites in the boundary diverges. 
Hence the connected regions in HaPPY \textit{almost} satisfy complementary recovery, making it a close model of holographic features for connected boundary regions.

\begin{figure}[t]
	\centering
	\subfloat[An example of a choice of connected region $R$ in a HaPPY code (shown in yellow) which obeys complementary recovery; i.e.\ its entanglement wedge is the complement of the entanglement wedge of its complement (shown in red). \label{fig:complementary}]{
	    \makebox[\linewidth][c]{
	    \includegraphics[width=0.5\linewidth]{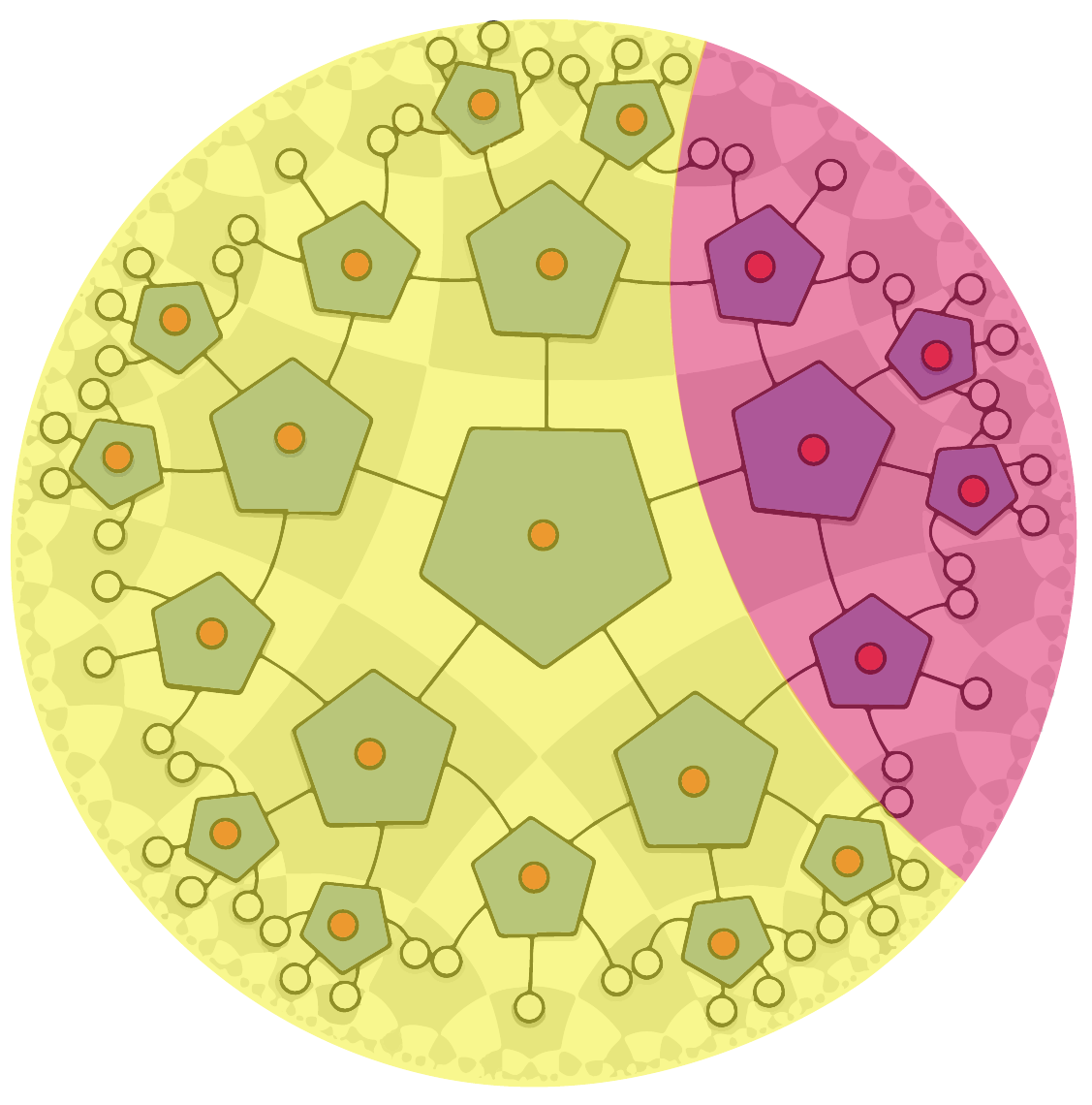}
	    }}
	\\
	\subfloat[A case in which a connected region does not obey exact complementary recovery, and instead a small $\mathcal{O}(1)$ residual region exists outside of both entanglement wedges. \label{fig:almostcomplementary}]{
	\makebox[\linewidth][c]{
	    \includegraphics[width=0.5\linewidth]{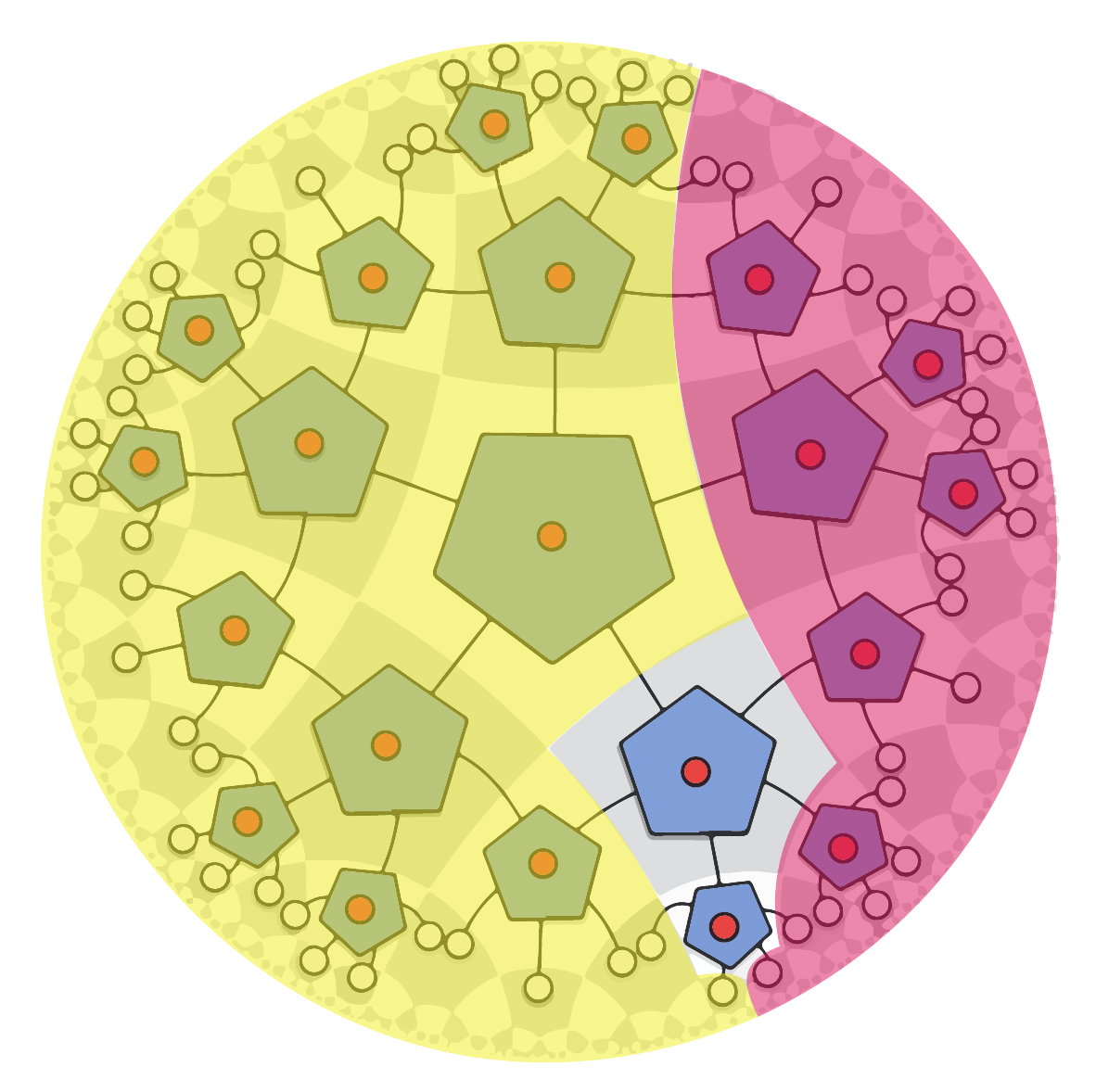}
	    }}
	\\
	\subfloat[An example of a choice of a highly disconnected region $R$ (shown along with its entanglement wedge in yellow) such that both it and its complement (shown along with its entanglement wedge in red) have empty entanglement wedges.
	This is an extreme case in which complementary recovery maximally fails to hold. \label{fig:empty}]{
	\makebox[\linewidth][c]{
	    \includegraphics[width=0.5\linewidth]{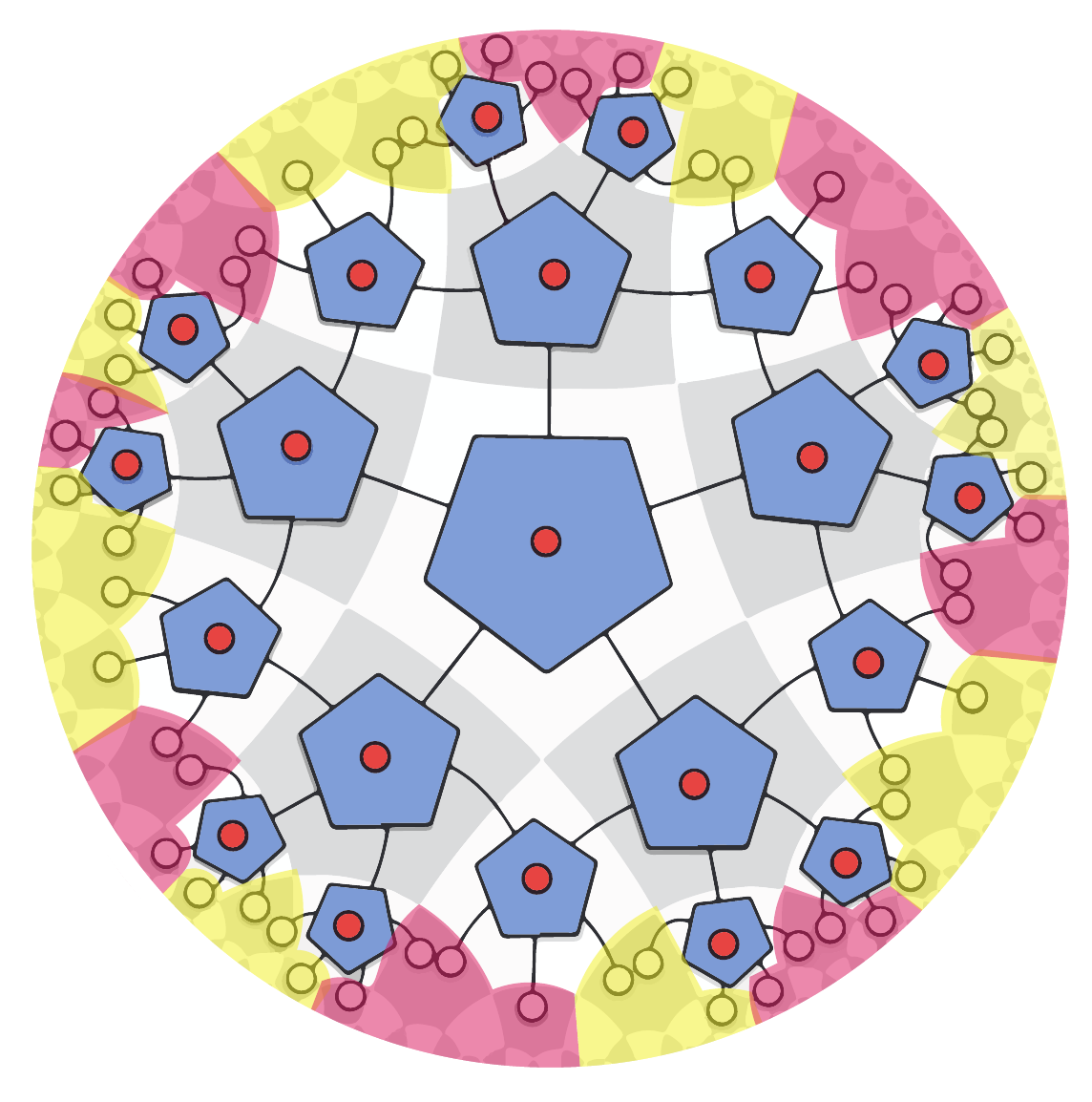}
	    }}
	\caption{Assorted examples where complementary recovery holds or fails in HaPPY. This and following figures are adapted from those in \cite{pastawskiHolographicQuantumErrorcorrecting2015}.}
	\label{fig:complementarity}
\end{figure}

\subsection{Geometric structure of holographic codes} \label{sect:geometry}
To this point, we have only relied on the tensor product structure of the physical and logical spaces, and not on the associated geometric structures that are typically associated with holographic codes (and assumed in pictures such as \cref{fig:complementarity}).
Although some of our results hold without a geometric structure, we require it at times -- for example, to generalize beyond transversal gates to locality-preserving gates in \cref{sect:nontransversal}.

Following Ref.~\cite{pastawskiCodePropertiesHolographic2017} each local encoded subsystem in the tensor product decomposition of \cref{eq:TPS} is associated with a specific point in some continuous bulk geometry.
A geometric region is then associated with the tensor product of Hilbert spaces corresponding to each local encoded subsystem located in that region.
We typically identify geometric distances in the bulk with the graph distances along the graph that underlies the tensor network of the encoding isometry.
The boundary geometry is then inherited by smoothly continuing this bulk geometry to the outward-facing legs of the outermost layer of tensors.
We make use of the geometric structure in two ways; first, we sometimes refer to \emph{connected} regions of the boundary, which are defined in terms of this geometry (e.g.\ in the above discussion about connected regions in HaPPY).
Second, we occasionally make use of distance measures in either the bulk or the boundary; these are defined in terms of the graph distance of the relevant lattice.
We assume these distance measures are normalized such that two adjacent sites are a distance of order one apart.

We now introduce some characterizations of how well protected a logical subsystem $S$ is against errors, some of which make explicit use of the geometric structure associated with holographic codes.
\begin{definition}\label{def:distance}
	We define the \textbf{\emph{code distance}} $d(S)$ of a logical subsystem $S$ in a subsystem code as the size of the smallest boundary region $R$ that is not correctable with respect to $S$, with size measured by number of qubits.
	This is equivalent to
	\begin{align}
	    d%_\sigma
		(S) = \min_R \{ %\sigma(R)
		|R|: R\text{ is not correctable wrt }S\}.
	\end{align}
\end{definition}
Physically motivated topological and holographic codes come endowed with a richer geometric structure.
This leads to refined notions of generalized code distance, several of which we now introduce with holographic codes in mind.
\begin{definition}\label{def:connectedDistance}
	We define the \textbf{\emph{connected code distance}} $d_c(S)$ of a logical subsystem $S$ of a subsystem code as the size of the smallest \emph{connected} boundary region $R$ that is non-correctable with respect to $S$, with size measured by number of qubits. 
	This definition satisfies $d_c \geq d$.
\end{definition}

Finally, we may wish to incorporate the geometric structure in any number of other ways.
We define a notion of distance which depends on an arbitrary ``size'' function, $\sigma(R)$, defined for any region $R$.
Typically, a useful definition will satisfy other constraints, like $R'\subset R \implies \sigma(R') < \sigma(R)$ or $\sigma(R)\geq 0$, but we do not require it.
\begin{definition}\label{def:generalDistance}
	Take any notion of ``size'' of a region, $\sigma(R)$.
	We define the \textbf{\emph{$\sigma$-code distance}} $d_\sigma(S)$ of a logical subsystem $S$ of a subsystem code as the minimum size of a non-correctable region, i.e.\
	\begin{align}
	    d_\sigma(S) = \min_R \{ \sigma(R) : R\text{ is non-correctable wrt }S\}.
	\end{align}
\end{definition}
We remark that \cref{def:connectedDistance} can be thought of as a special case, in which $\sigma$ is defined to be infinity for non-connected regions, and the number of qubits in the region otherwise.

Along with the above definitions of distance, we recall the notion of \emph{code price}~\cite{pastawskiCodePropertiesHolographic2017}, defined as follows.
\begin{definition}
	We define the \textbf{\emph{%$\sigma$-
	code price}} $p(S)$ of a logical subsystem $S$ of a subsystem code as the minimum size of a region which can reconstruct all logical operators on $S$%.
	, with size measured by number of qubits.
	This is equivalent to
	\begin{align}
	    p%_\sigma
		(S) = \min_R \{ %\sigma(R)
		|R|: R^c\text{ is correctable wrt }S\}.
	\end{align}
\end{definition}
As with distance, equivalent versions of the connected code price $p_c$ and the $\sigma $-code price $p_{\sigma }$ can be defined.
We remark that the following is a straightforward corollary of \cref{eq:ewlimit}.
\begin{lemma}
	The price is larger than the distance, $p(S) \geq d(S)$.
	\label{lem:pricedistance}
\end{lemma}

We expect the entanglement wedge map in a holographic stabilizer code to be compatible with its geometric structure, as is the case in AdS/CFT. 
Specifically, we expect that a local extension of the bulk region should be possible by locally deforming the boundary region -- the entanglement wedge map is ``smooth'' in some sense%
\footnote{There can be discontinuities in the entanglement wedge map, such as when a phase transition occurs between connected and disconnected entanglement wedges.
These do not contradict our definition of smoothness.}%
.
Of course, the hyperbolic nature of geometries associated with spatial slices of AdS means that an $\mathcal{O}(1)$ extension of the entanglement wedge in the bulk by an amount $\delta$ may require an $\mathcal{O}(\exp(r+\delta))$ extension to the diameter of the boundary region, where $r$ is the distance to the boundary in AdS units (the continuous analogue of the number of layers in HaPPY).
Nevertheless, this weak notion of continuity is sufficient for our purposes, as the distance of a bulk region typically also scales exponentially with $r$.
The following definition formalizes this notion.

\begin{definition}
    A code is said to have a \emph{$\kappa_R$-smooth entanglement wedge map} if for any choice of a connected boundary region $R$ and positive numbers $\delta,x$ that satisfy $\kappa_R(\delta) \le x$ then there exists a connected boundary region $R_{x}$ such that $R\subset R_x$, $\abs{R_x -R}\le x$ and $B(\mathcal{E}(R),\delta) \subset\mathcal{E}(R_x)$, where $B(\mathcal{E}(R),\delta)$ is the region of all points with distance at most $\delta$ from $\mathcal{E}[R]$.
    In most cases we take $\kappa_R$ to be a non-decreasing function that satisfies $\kappa_R(0)=0$. 
	\label{def:kappa}
\end{definition}

In the above definition $\kappa_R$ measures how much a boundary region $R$ must be extended to accommodate a small extension of its entanglement wedge $\mathcal{E}(R)$ in the bulk. 

\begin{figure}[t]
	\centering
	\subfloat[An example of a region $R$ with small $\kappa _R$. 
		The region $R$ is shown in red, and for an expansion of $1$ graph distance in the bulk, a boundary expansion of only $1$ site in either direction is required; thus $\kappa _R(1) = 1$. The expanded region is shown in yellow. \label{fig:kappa1}]{
		\makebox[\linewidth][c]{
		\includegraphics[width=.5\linewidth]{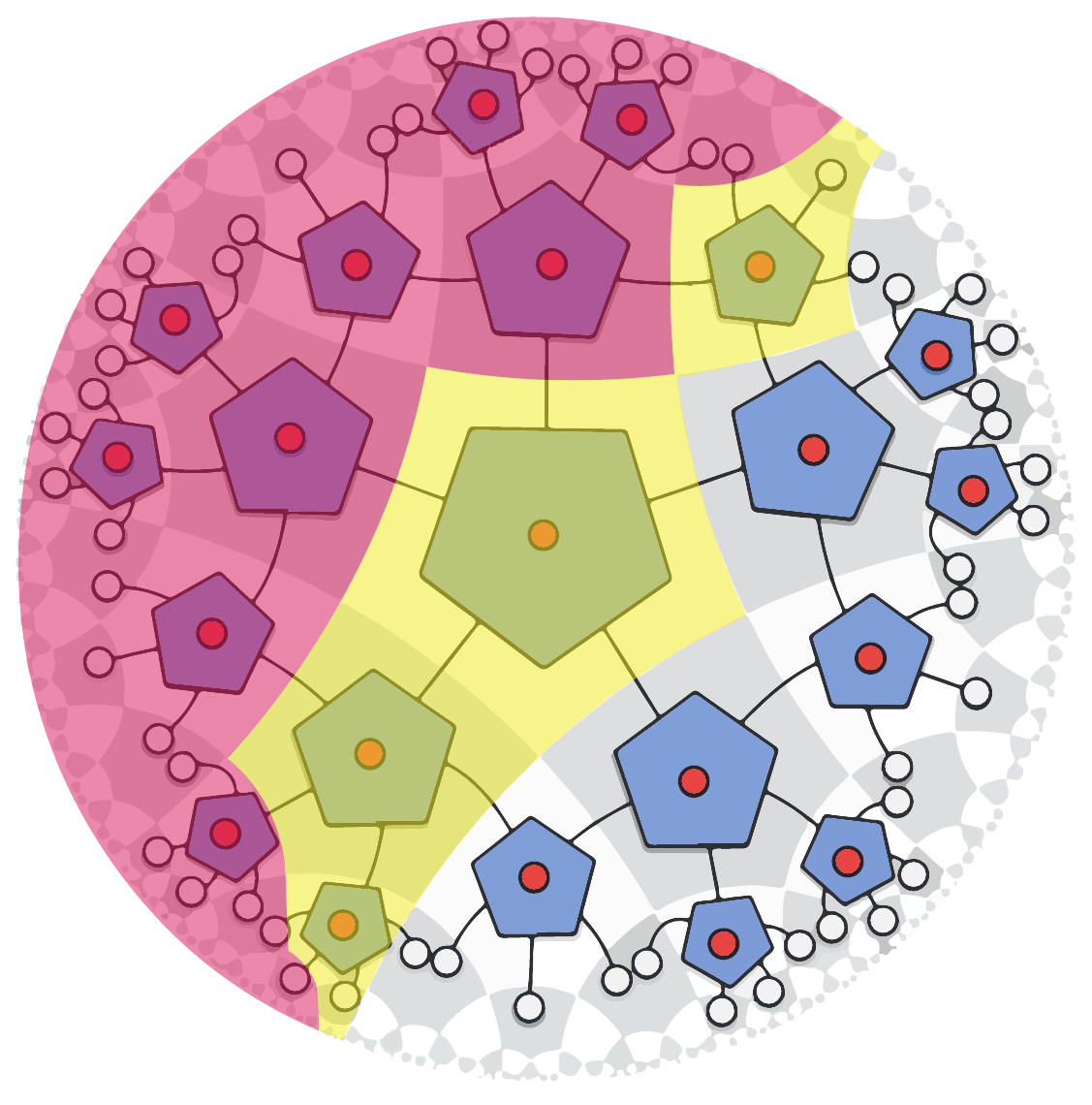}
		}}
		\\
		\subfloat[An example of a region $R$ with a large $\kappa _R$. 
		Expansion by $1$ in the bulk requires at least $11$ sites added to either side, thus $\kappa _R(1) = 11$. \label{fig:kappa2}]{
		\makebox[\linewidth][c]{
		\includegraphics[width=.5\linewidth]{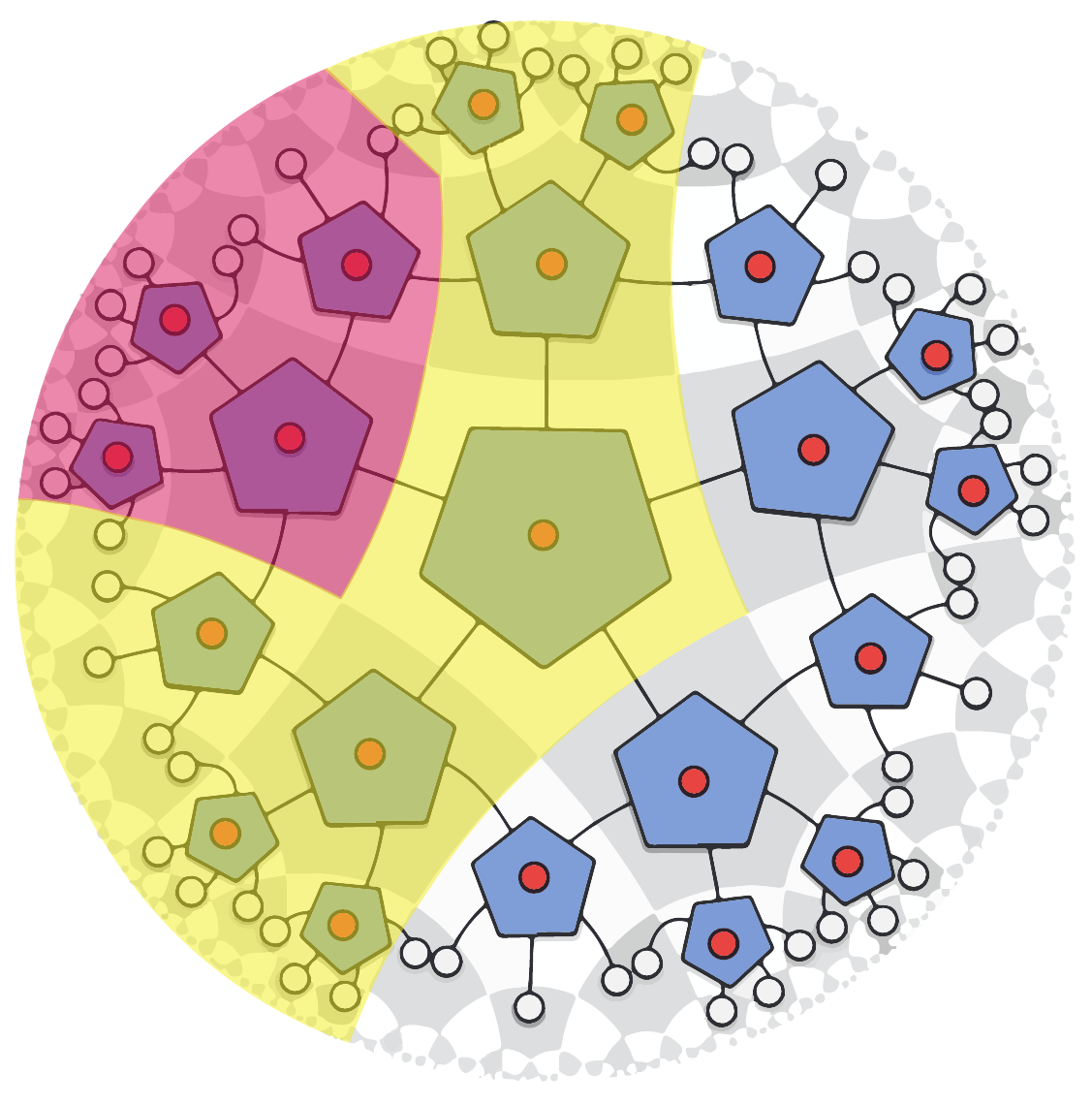}
		}}
	\caption{Illustration of the idea of a $\kappa_R $-smooth entanglement wedge map, as defined in \cref{def:kappa}.}
	\label{fig:kappa}
\end{figure}

\subsection{Stabilizer subsystem codes}
\label{subsec:stabilizer-subsystem-codes}
The main results of this paper refer to the Clifford hierarchy, which is closely related to the Pauli group that arises in stabilizer codes.
In this section we carefully introduce this notion to point out a few subtleties.

Given a decomposition of a Hilbert space into qudits and a particular choice of basis for each qudit, one can construct a Pauli group as follows.
The Hilbert space decomposes into $d$-dimensional qudits, $\mathcal{H} = \bigotimes_{n=1}^{N} \mathcal{H}_n$, with a specific basis $\left\{ \ket{j}_n \right\}_{j=1}^{d}$ chosen for each $\mathcal{H}_n \cong \mathbb{C}^d$.
The Pauli group $\mathcal{P} $ is generated by generalized $X$ operators $X_n = \sum_{j=1}^{d} \ket{(j+1) \textrm{ mod }d}_n\!\bra{j} $  (with identities acting on all but the $n$th qudit) and generalized $Z$ operators $Z_n= \sum_{j=1}^{d} \omega _d^j \ket{j}_n\!\bra{j}$ (with $\omega _d $ a primitive $d$th root of unity). 

A \emph{stabilizer} code requires the physical space to be equipped with a specific tensor product structure and choice of local basis, and some abelian subgroup of the Pauli group $\mathcal{S}\subset \mathcal{P}$ such that the only element of $\mathcal{S}$ in the center\footnote{Recall that the center $\mathcal{Z}[A]$ of an algebra $A$ is the set of elements of $A$ that commute with all of $A$.} of $\mathcal{P}$ is the identity, $\mathcal{Z}[\mathcal{P}] \cap \mathcal{S} = \left\{ \Id \right\}$\footnote{In other words, for $j=1,\dots,d-1$, $ \omega_d^j \notin \mathcal{S}$. In the qubit case, this condition coincides with the familiar $-\Id \notin \mathcal{S}$.}; this group $\mathcal{S}$ is called the \emph{stabilizer} group.
The code subspace $\mathcal{H}_c$ is then characterized as the space of states that are $+1$-eigenstates of all of elements $S\in \mathcal{S}$.
For a given choice of stabilizer group, there exists a unique\footnote{The tensor product structure and local basis for the encoded space $\mathcal{H} = \mathcal{H}_L \otimes \mathcal{H}_J$ are only unique up to transformation by a Clifford operation, which preserves the encoded Pauli group.} tensor product structure and local basis for the encoded space $ \mathcal{H}_e \cong \mathcal{H}_c$ such that the codespace-preserving physical Pauli operators implement elements of the encoded Pauli group \cite{nielsenQuantumComputationQuantum2010}.
A \emph{stabilizer subsystem} code corresponds to a particular division of these encoded qudits into logical and junk subsystems such that the logical subsystem has at least one qudit\footnote{One can equivalently define a stabilizer subsystem code by choosing some generally non-abelian subgroup $\mathcal{G}\subset \mathcal{P}$, defining the stabilizer group $\mathcal{S}$ as the center of $\mathcal{G}$, and then defining the division into logical and junk subsystems such that elements of the quotient group $\mathcal{G}/\mathcal{S}$ are identified with dressed-CSP operators that act only on the junk subsystem~\cite{Poulin2005,Bacon2005a}.}.

When the junk subsystem is not trivial, it may be possible that no operator supported on $R^c$ can implement the logical Pauli operator $A_L \otimes \Id_J$, but some \emph{dressed-CSP} operator can implement $A_L \otimes A_J$.
Thus it is useful to introduce a weaker correctability property that we term \emph{dressed-cleanable}, following Ref.~\cite{pastawskiFaulttolerant2015}. 

\begin{definition}
	$R$ is \textbf{\emph{dressed-cleanable}} if for any logical Pauli operator $P_L$ there exists a dressed-CSP operator supported on $R^c$ that implements $P_L \otimes P_J$, for some operator $P_J \in \mathcal{B}(\mathcal{H}_J)$.
	\label{def:dressedCleanable}
\end{definition}
This is a strictly weaker requirement than correctability. Compare both notions using tensor networks in Fig.~\ref{fig:cleanable-tn}. 
\begin{lemma}
	A correctable region is also dressed-cleanable.
	\label{lem:BCDC}
\end{lemma}
\begin{proof}
	Suppose $R$ is correctable.
	Then for each bare-logical operator, such as $P_L$, there exists a bare-CSP operator supported in $R^c$ that implements it.
	This is also a dressed-CSP operator, satisfying the condition for dressed-cleanability.
\end{proof}

We remark again that due to the dependence on dressed-CSP operators, \cref{def:dressedCleanable} is defined relative to a particular choice of logical subsystem.
Similar to \cref{lem:dressedtrivial}, a corresponding property holds for dressed-cleanable regions.

\begin{figure}[t]
	\centering
	\subfloat[Tensor network depiction of the definition of a bare cleanable region $R$.]{
	\makebox[\linewidth][c]{
	\includegraphics[width=1.0\linewidth]{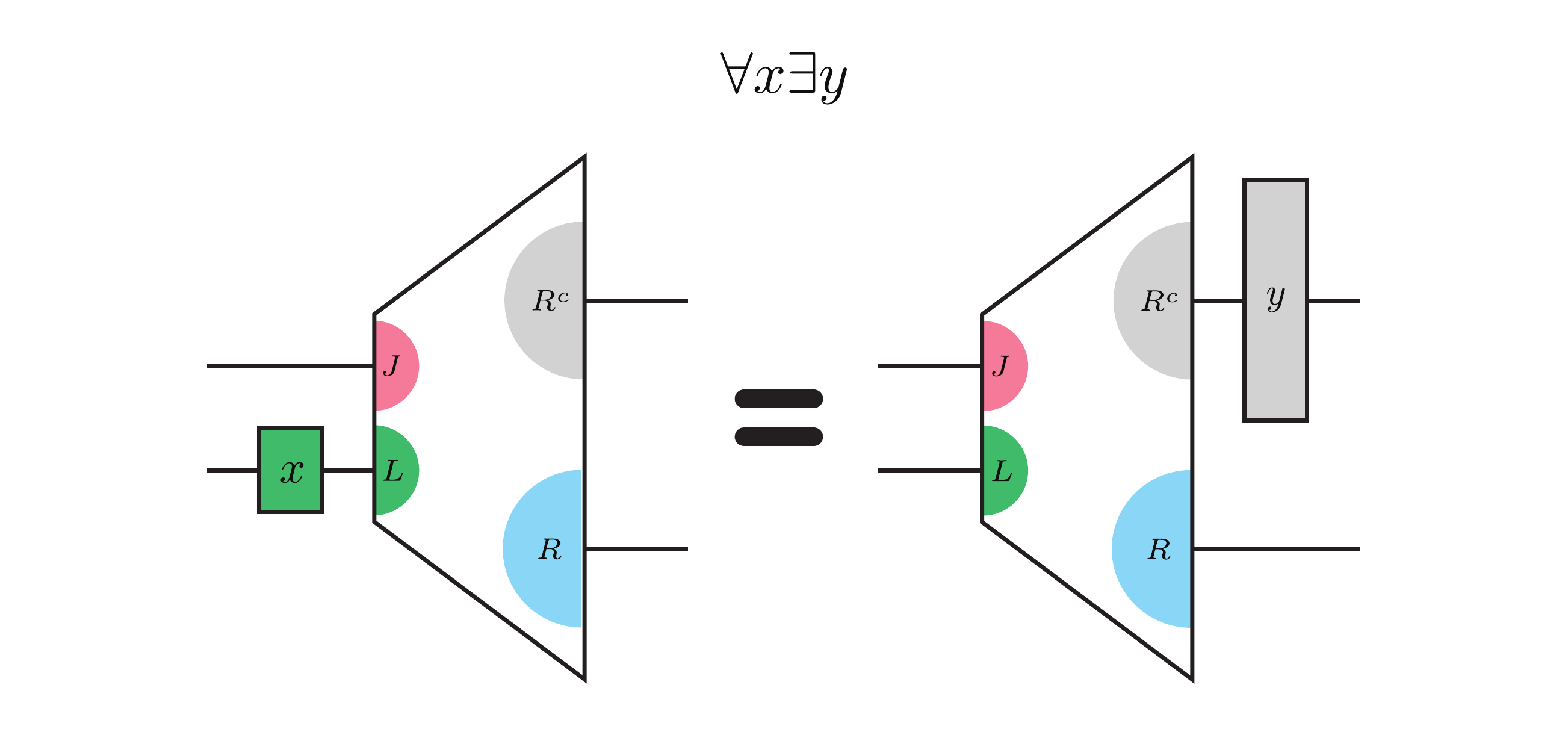}
	}}
	\\
	\subfloat[Tensor network depiction of the definition of a dressed cleanable region $R$]{
	\makebox[\linewidth][c]{
	\includegraphics[width=1.0\linewidth]{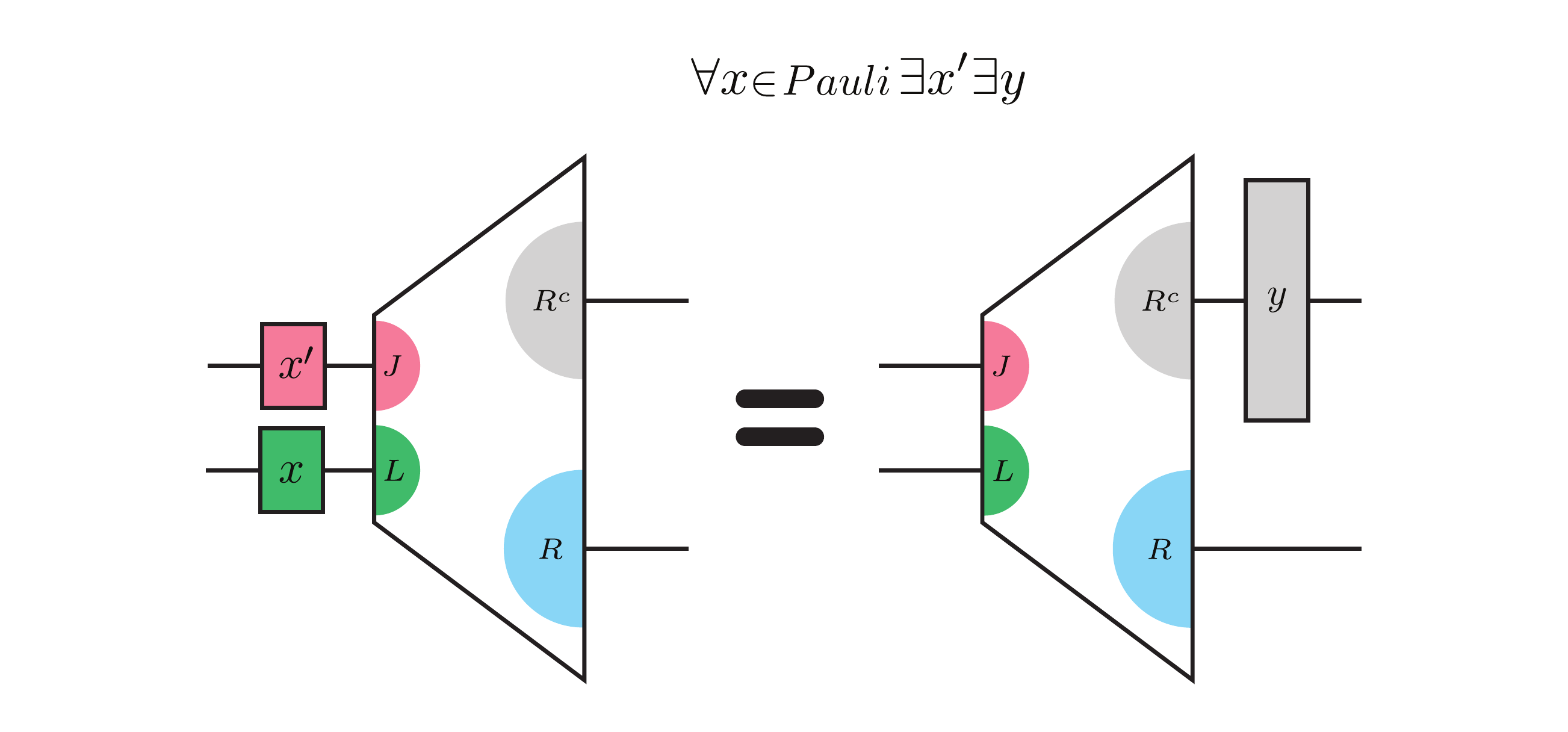}
	}}
\caption{Bare and dressed cleanability in tensor network notation. In the above, $L$ denotes the protected logical subsystem, while $J$ denotes the junk subsystem. }
\label{fig:cleanable-tn}
\end{figure}

\begin{lemma}
	Suppose a region $R$ is \emph{dressed-cleanable}.
	Then any \emph{bare-CSP} operator supported on $R$ implements a logical operation of the form $c\Id_L \otimes \Id_J$, with $c\in \mathbb{C}$.
	\label{lem:baretrivial}
\end{lemma}
\begin{proof}
Suppose $R$ is dressed-cleanable.
Then take any bare-CSP operator $A$ supported on $R$ with $V^{\dagger} A V = A_L \otimes \Id_J$.
By assumption, for any $P_L \in \mathcal{P}(\mathcal{H}_L)$, one can find $P$ supported on $R^c$ such that $V^{\dagger} P V = P_L \otimes P_J$.
Since $A$ and $P$ are supported on complementary regions, we have $[A,P]=0$, which implies $[A_L, P_L ] = 0$.
Since this is true for any $P_L$, $A_L$ must commute with all Pauli operators in $\mathcal{P}(\mathcal{H}_L)$; thus it commutes with all of $\mathcal{B}(\mathcal{H}_L)$, so it must be proportional to the identity, $A_L = c\Id_L$.
\end{proof}

Furthermore, although dressed-cleanability is generally a weaker requirement than correctability, these properties are in fact equivalent for a region obeying complementary recovery.
\begin{lemma}
	Suppose that some region $R$ obeys \emph{complementary recovery}, and it is dressed-cleanable with respect to the logical subsystem given by $S$.
	Then it is also correctable with respect to $S$.
	\label{lem:CR-DC}
\end{lemma}
\begin{proof}
	By assumption, $\mathcal{E}[R^c] \cup \mathcal{E}[R] = \left\{ 1,\dots ,n_L \right\}$, the full set of all local encoded subsystems; furthermore $R$ is correctable with respect to $\mathcal{E}[R^c]$ and $R^c$ is correctable with respect to $\mathcal{E}[R]$.
	If any local encoded subsystems of $S$ are in $\mathcal{E}[R]$, then correctability of $R^c$ ensures that logical operators on those local encoded subsystems can be implemented on $R$, which would be non-trivial bare-CSP operators; contradicting \cref{lem:baretrivial}.
	Thus $S \cap \mathcal{E}[R] = \emptyset$, and therefore $S \subset \mathcal{E}[R^c]$.
	Since $R$ is correctable with respect to $\mathcal{E}[R^c]$, it is also correctable with respect to $S$.
\end{proof}

In stabilizer subsystem codes, it has been shown that the converse directions for \cref{lem:baretrivial,lem:dressedtrivial} are also true (see Ref.~\cite{bravyi_SubsystemLocal}, which is easily generalized from qubit to qudit stabilizer codes).
This proof is more involved, however, and fortunately we do not need it here.
There is more to be said about a number of different correctability properties of regions of subsystem codes, which we elaborate on further in \cref{app:regions} for the interested reader.

\subsection{Fault-tolerant operations in stabilizer subsystem codes}\label{sect:fault-tolerant}

A core result from the theory of quantum error-correcting codes is the Eastin-Knill theorem~\cite{eastinRestrictionsTransversalEncoded2009}, which states that no code can allow a universal set of gates to be implemented transversally, as long as it can detect errors on a single subsystem.
Several related results have strengthened this in various contexts~\cite{bravyiClassificationTopologicallyProtected2013,pastawskiFaulttolerant2015,beverlandProtectedGatesTopological2016,Jochym-OConnor2017,Webster2018,faistContinuousSymmetriesApproximate2019,Woods2020,Burton2020}, such as more narrowly specifying the set of gates that \emph{can} be implemented transversally.
For a family of stabilizer subsystem codes with locally generated stabilizer groups known as \emph{topological stabilizer subsystem codes}, the Pastawski-Yoshida theorem~\cite{pastawskiFaulttolerant2015} 
establishes a restriction on transversally implementable gates 
that depends only on the dimension of the code's geometry.
Specifically, in $D$ dimensions a transversal gate must implement a logical operation to the $D$th level of the so-called \emph{Clifford hierarchy}, which is defined as follows.
\begin{definition}
Denote $\mathcal{P}$ as the Pauli group on the logical subsystem $\mathcal{H}_L$, and the zeroth level of the Clifford hierarchy to be $\mathcal{C}_0 \defi \mathbb{C}$.
Then the Clifford hierarchy for $i>0$ is a family of sets $\mathcal{C}_i$ of unitary operators defined recursively as follows%
\footnote{We remark that this differs slightly from the standard definition of the Clifford hierarchy \cite{bravyiClassificationTopologicallyProtected2013}, which consists of unitaries such that $UPU^\dagger \in \mathcal{C}_{i-1}$.
It was shown in \cite{pastawskiFaulttolerant2015} that the definitions coincide for $i>1$.
For $i=1$, this definition gives $\mathcal{C}_1 = \mathcal{P} \times \mathbb{C}$, whereas the standard definition gives $\mathcal{C}_1 = \mathcal{P}$.
This definition also allows $\mathcal{C}_0$ to be non-empty, unlike the standard definition.}%
:
\begin{align}
    \mathcal{C}_i = \{ U | \forall P \in \mathcal{C}_1, UPU^\dagger P^\dagger\in \mathcal{C}_{i-1} \}
\end{align}
\end{definition}
We emphasize that this is defined in terms of the Pauli strings acting on the logical subsystem $\mathcal{H}_L$.

In any code with a finite distance, another set of fault-tolerant gates is those that do not increase the support of an operator too much under conjugation.
For example, suppose that an operator $U$ has the property that when conjugating any local one-site operator $X$, the resultant operator $UXU^\dagger$ has support on only $k<d$ sites. 
Then any one-site error will remain correctable after the application of this gate.
However, successive applications of these gates may expand the support of the error exponentially, unless there is some particular structure in \emph{how} the errors are spread (or $k=1$, for example with transversal gates).

Fortunately, when there is a geometric structure associated with physical space, such as in topological codes and holographic codes, we can restrict to a more specific class of gates with even better protection properties.
If application of a unitary $U$ causes the support of an error to spread by some bounded amount, not just in terms of the number of sites affected, but also in terms of the \emph{geometric extent} of these errors, then we say that $U$ is \emph{locality-preserving}.
For example, evolution by a local Hamiltonian for a short time period can be well-approximated by such a unitary, e.g.\ via Trotterization.
We remark that unlike the previous example, successive applications of locality-preserving gates in flat geometries only spread errors to a polynomial number of total sites in the number of applications (raised to the power $D$, the geometric dimension), rather than an exponential number.

We now formalize this notion.
\begin{definition}
The extent to which a gate $U$ is locality-preserving is quantified by its \emph{spread} $s_U$, the maximum distance by which it can increase the support of a local operator supported on some region $R$.
Formally, the spread of $U$ is 
\begin{align}
	s_U = \min \{ d\in [0,\infty] | \forall \text{ regions }R, \forall X \in \supp(R), \ \
	\nonumber \\
	UXU^{\dagger} \in \supp(B(R,d))  \}.
\end{align}
\end{definition}
We remark that the spread of a transversal operator is zero.

It was shown by Pastawski and Yoshida in Ref.~\cite{pastawskiFaulttolerant2015} that for any $D$-dimensional topological subsystem stabilizer code satisfying some basic error-correction properties\footnote{In particular, a logarithmic distance and a nonzero error-threshold.}, any dressed-CSP operator with sufficiently limited spread can only implement a dressed-logical operator $A_L \otimes A_J$ such that $A_L\in \mathcal{C}_D$.
This generalizes a result due to Bravyi and Koenig \cite{bravyiClassificationTopologicallyProtected2013} that applies only to standard topological stabilizer codes.
A key step in the proof of this result was the following lemma, which we make frequent use of in this work.
\begin{lemma}[Pastawski-Yoshida lemma]
	For a subsystem stabilizer code, let $U$ be a dressed-CSP unitary operator supported on the union of regions $\left\{ R_j \right\}_{j\in[0,m]}$. If $R_0$ is correctable and each ${R_j^+ \defi B(R_j,2^{m-j}s_U)}$ is dressed-cleanable for $j>0$, then the logical unitary implemented by $U$ belongs to $\mathcal{C}_m$.
	\label{lem:PY}
\end{lemma}
For completeness, we replicate the proof of this lemma here from Ref.~\cite{pastawskiFaulttolerant2015}.
\begin{proof}
    Use induction on $m$.
    For $m=0$, $U$ is a dressed-CSP operator supported on a correctable region, so by \cref{lem:dressedtrivial}, it must implement a scalar multiple of the identity on the logical subsystem, i.e.\ an element of $\mathcal{C}_0$, thus the base case holds.
    
    Now assume the case $m$ holds, and that we wish to show it for $m+1$.
    Suppose we have a dressed-CSP unitary $U$ supported on $\bigcup_{j=0}^{m+1} R_j$ such that ${R_j^+ \defi B(R_j,2^{m+1-j}s_U)}$ is dressed-cleanable for each $j>0$.
    Consider an arbitrary Pauli operator on the logical subsystem, $P_L\in \mathcal{P}$. 
    Because $R_1^+$ is dressed-cleanable, $P$ can be implemented by a dressed-CSP operator $P$ supported in $\supp(P) \subseteq  (R_{m+1}^+)^c$.
    Then
    \begin{align}
        \supp(UPU^\dagger P^\dagger) &\subseteq \supp(UPU^\dagger) \cap \supp(PU^\dagger P^\dagger) \\
        &\subseteq B(\supp(P), s_U) \cap \supp(U) \\
        &\subseteq B((R_{m+1}^+)^c, s_U) \cap \bigcup_{j=0}^{m+1} R_j \\
        &\subseteq R_{m+1}^c \cap \bigcup_{j=0}^{m+1} R_j \\
        &\subseteq \bigcup_{j=0}^{m} R_j,
    \end{align}
    where we used the fact that $R_{m+1}^+ = B(R_{m+1},s_U)$ in the second last line.
    We can now apply the inductive hypothesis to the dressed-CSP unitary $UPU^\dagger P^\dagger$, so long as $B(R_j,2^{m-j} s_{UPU^\dagger P^\dagger}) $ are dressed-cleanable.
    Fortunately, since $s_{UPU^\dagger P^\dagger} \leq 2s_U$, this is implied by the fact that $B(R_j,2^{m+1-j}s_U)$ is dressed-cleanable, and thus $UPU^\dagger P^\dagger \in \mathcal{C}_m$.
    By the definition of the Clifford hierarchy, we have $U\in \mathcal{C}_{m+1}$.
\end{proof}
This lemma is simpler in the case that $U$ is transversal, which implies that $s_U = 0$ and so $R_j^+ = R_j$.
We also introduce a stronger form of this lemma in the case that $U$ is also bare-CSP.
We prove it using virtually an identical proof to that of \cref{lem:PY} from Ref.~\cite{pastawskiFaulttolerant2015}.
\begin{lemma}
	For a subsystem stabilizer code, let $U$ be a \emph{bare}-CSP unitary operator supported on the union of regions $\left\{ R_j \right\}_{j\in[0,m]}$. If $R_0$ and each $R_j^+ \defi B(R_j,2^{j-1}s_U)$ are dressed-cleanable, then the logical unitary implemented by $U$ belongs to $\mathcal{C}_m$. 
	\label{lem:PYbare}
\end{lemma}

\begin{proof}
    Use induction on $m$.
    For $m=0$, $U$ is a bare-CSP operator supported on a dressed-cleanable region, so by \cref{lem:baretrivial}, it must implement a scalar multiple of the identity on the logical subsystem, i.e.\ an element of $\mathcal{C}_0$, thus the base case holds.
    
    The proof then follows identically to \cref{lem:PY}, until it comes time to apply the inductive hypothesis, which now additionally requires that $UPU^\dagger P^\dagger$ be a bare-CSP operator.
    Fortunately, as $U$ is a bare-CSP operator, then even though $P$ is only dressed-CSP, the component $P_J$ acting on the junk subsystem cancels with its inverse, and the product is indeed a bare-CSP operator.
\end{proof}

The final lemma of this section establishes a useful sufficient condition, in terms of price, to guarantee the existence of three correctable regions that suffice to place restrictions on the logical action of fault-tolerant gates.

\begin{lemma}
	If a logical subsystem $S$ of a subsystem stabilizer code satisfies the relationship
	\begin{align}
		2 \leq p(S) \leq 2 d(S) - 2,
		\label{eq:pricedistance}
	\end{align}
	then any dressed-CSP transversal unitary $U$ can only implement an element of the Clifford group $\mathcal{C}_2$ on the logical subsystem associated with $S$.
	\label{lem:PYprice}
\end{lemma}
\begin{proof}
	By definition of price, there exists some region $R$ with size $|R| = p(S)$ such that $R^c$ is correctable.
	Then as $|R|\geq 2$, $R$ can be split into two non-empty regions $R_1$ and $R_2$ such that $|R_1|, |R_2| \leq d(S)-1$, meaning that $R_1$ and $R_2$ are smaller than the smallest non-correctable region; thus they must be correctable.
	Also, recall that \cref{lem:BCDC} showed that correctability is stronger than dressed-cleanability, so these regions are also dressed-cleanable.
	Thus we can apply \cref{lem:PY} with $s_U = 0$ because of its transversality, to conclude that $U$ implements a Clifford gate, $U_L \in \mathcal{C}_2$.
\end{proof}
We re-express the lemma in this way in order to crystallize the physical intuition for why holographic codes should not permit transversal non-Clifford gates for sufficiently small bulk regions.
In a code satisfying complementary recovery, it was shown in Ref.~\cite{pastawskiCodePropertiesHolographic2017} that the distance is equal to the price for a logical subsystem associated with a single point.
Although holographic stabilizer codes do not generally satisfy complementary recovery, we use this observation as motivation to argue that the difference between price and distance is an effective measure of the effective ``size'' of a region, after accounting for the degree to which complementary recovery fails.
In this way, \cref{lem:PYprice} can be interpreted as saying non-Cliffords cannot be transversally implemented on sufficiently small regions -- i.e. those with $p-d \leq d - 2$.

We remark that an analogous definition can be formed in terms of the generalized notions of price and distance, $p_{\sigma }(S)$ and $d_\sigma (S)$.
Finally, the lemma could also be strengthened by redefining distance in terms of the smallest \emph{non-dressed-cleanable} region; this would be a larger notion of distance and thus easier to satisfy \cref{eq:pricedistance}.
$R_1$ and $R_2$ would be dressed-cleanable, and thus \cref{lem:PY} could still be invoked.

We remark that both \cref{lem:PY,lem:PYbare} were originally proven in the context of qubit stabilizer codes, but easily generalize to qudits.
From here on, we restrict our exposition to the case of qubits for simplicity, however we emphasize that all of our results apply likewise to qudit-based holographic stabilizer codes.

\section{HaPPY codes \& non-Clifford gates}\label{sect:happy}

The simplest holographic stabilizer code, commonly referred to as HaPPY \cite{pastawskiHolographicQuantumErrorcorrecting2015}, is built out of \emph{perfect tensors} -- a class of tensors with particularly nice error-correction properties that allow the entanglement wedge map to mimic the geometric behaviour of AdS/CFT.
As the encoding map is given by a tensor network built from identical local encoding tensors, 
any gates that can be implemented transversally on this encoding tensor are inherited by the larger network\footnote{
This is a specific instance of global symmetry inherited from local symmetry, which is common in the tensor network literature~\cite{1367-2630-12-2-025010,Singh2010,Singh2013,williamson2014matrix,Bridgeman2017,NewSETPaper2017}, an interesting open question is whether a converse statement of global implying local tensor symmetry can be proven for holographic tensor networks.}. 
For example, the standard choice of perfect tensor is the encoding isometry of the five-qubit code (see the appendix of Ref.~\cite{pastawskiHolographicQuantumErrorcorrecting2015} for a review). This code admits a transversal gate set generated by the homogeneous tensor product of single site Pauli group elements and the Clifford gate $K$, which can be defined as mapping $X \to Z \to Y \to X$ under conjugation \cite{Yoder2016}.
Transversal implementation of the group generated by these gates is inherited by the larger network built out of this perfect tensor, as shown in \cref{fig:transversalpushing}. 
This is the most straightforward way in which a holographic code can allow for transversal implementation of a logical gate; and we now argue that no non-Clifford gate can be implemented in such a fashion.

Consider a rank-$2N$ tensor $T_{i_1,i_2,\dots ,i_{2N}}$, and suppose that for any $n\leq N$, the matrix given by lowering the first $n$ indices is an isometry, i.e.\
\begin{align}
	T_{i_1,\dots, i_n}^{i_{n+1},\dots ,i_{2N}} T_{ i_{n+1},\dots ,i_{2N}}^{\dagger j_1,\dots, j_n} = \delta _{i_1,\dots, i_n}^{j_1,\dots ,j_n}.
\end{align}
If this holds not only for the first $n$ indices but for \emph{any} choice of $n$ indices, then we say that $T$ is a perfect tensor.
We remark that the tensor need not be permutation invariant, i.e.\ it may form different isometries depending on which legs are chosen to be inputs and which to be outputs.
For a review of perfect tensors, see Ref.~\cite{pastawskiHolographicQuantumErrorcorrecting2015}.

Such a tensor can be interpreted as a good quantum error-correcting code with parameters $\left[ \left[ n,k,d \right] \right] = \left[ \left[ 2N-1,1,N \right] \right]$,  
by lowering only the first index and treating this as the encoding isometry:
\begin{align}
	V\defi T_{i_1}^{i_2,\dots ,i_{2N}} : \mathcal{H}_L = \mathcal{H}_2 \to \mathcal{H}_P = \mathcal{H}_2^{\otimes (2N-1)},
\end{align}
with the distance of $N$ being guaranteed by the perfect tensor property (see Ref.~\cite{pastawskiHolographicQuantumErrorcorrecting2015} for details).
Furthermore, it guarantees that any region of $N-1$ qubits or fewer is correctable.
Thus, the $2N-1$ qubits in the physical system can easily be partitioned into three sets such that each is correctable.
By \cref{lem:PY}, this means that for any transversal unitary $U$, the logically implemented gate is a Clifford, $V^{\dagger}UV \in \mathcal{C}_2$.

\begin{figure}[t]
	\centering
	\includegraphics[width=\linewidth]{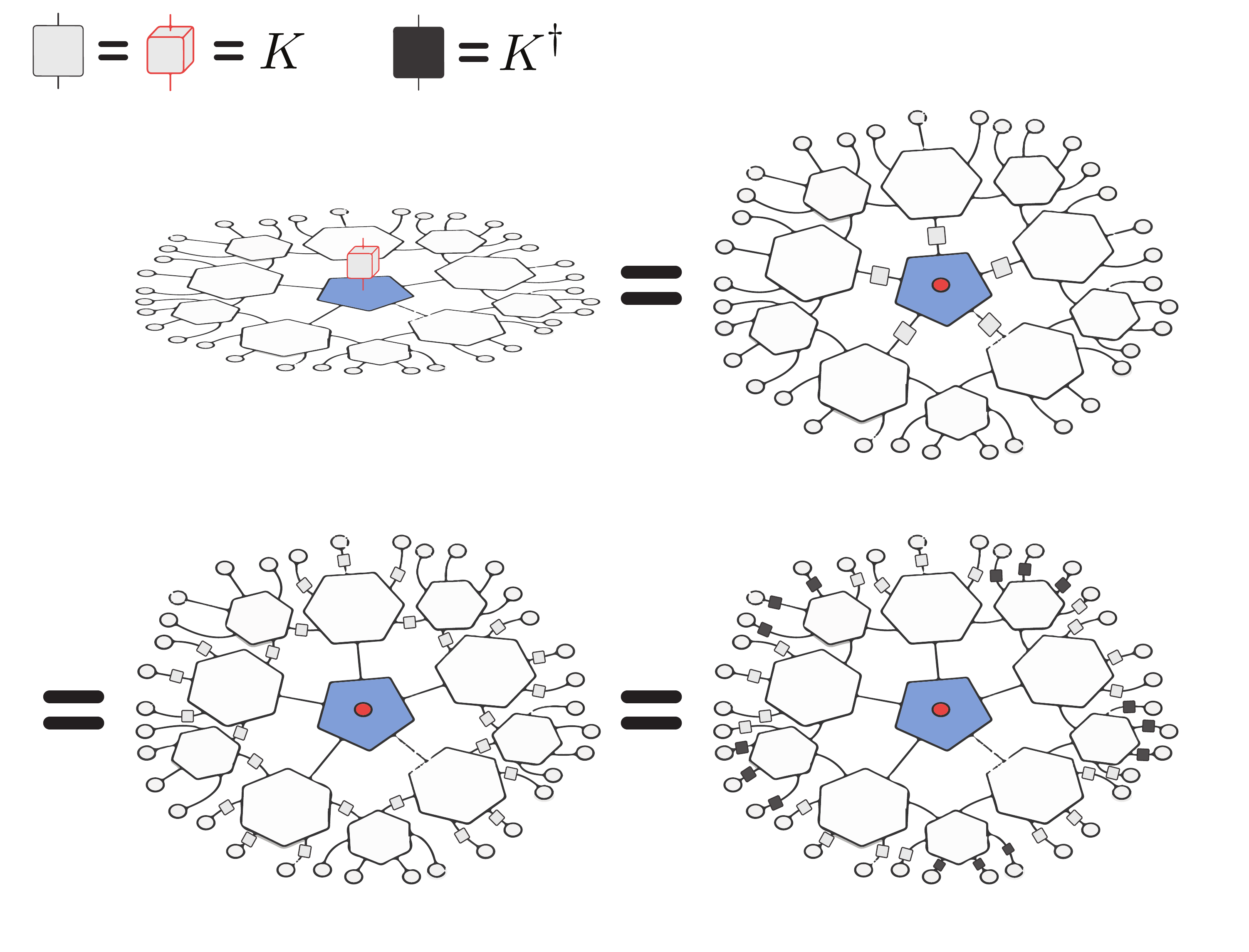}
	\caption{Pictorial argument that transversal implementation of a gate on a single perfect tensor implies transversal implementation for the entire HaPPY network.}
	\label{fig:transversalpushing}
\end{figure}

Furthermore, the same can be said about holographic codes built from a slightly broader class of tensors known as \emph{planar perfect} \cite{berger2018perfect} or \emph{block perfect} \cite{harrisCalderbankSteaneShorHolographicQuantum2018}.
In these tensors, the order of the indices matters -- the isometry property is only required to hold for a set of $n$ consecutive indices (where $i_{2N}$ is understood to be adjacent to $i_1$).

Of course, the above argument does not rule out the possibility of transversally implemented non-Clifford gates in networks built out of these tensors -- it only excludes the simplest possible way for such a gate to be implemented, which would be that it is transversal on each individual tensor that makes up the network.
Nevertheless, rigorous arguments can be formulated that rule out any kind of transversal non-Clifford in a HaPPY code.
Essentially, one starts with a logical operator on some central tensor and considers the boundary regions that it can be pushed out to at each successive layer of the network.
Due to the ``negative curvature'' (which is a necessary requirement for the network to form a bulk to boundary isometry), the angular size of the region which can be pushed to does not increase too much layer by layer; the total size of each layer grows at least as quickly as the region being pushed to does.
Thus the complement of the region being pushed to, which is correctable, is sufficiently large that one can divide the physical space into three such regions and proceed by applying \cref{lem:PY} to exclude the possibility of transversal non-Cliffords. 

In \cref{app:schlafli}, we use a rigorous form of the above argument to prove that any HaPPY network built out of a homogeneous hyperbolic planar tiling does not accommodate the transversal application of non-Cliffords. 
We remark that the idea behind our proof of no transversal non-Clifford gates in holographic tensor network stabilizer codes can be extended to rule out more general locality preserving gates. In the following section we describe such a generalization for more general families of sufficiently holographic stabilizer codes that should include tensor network codes as a special case. 

\section{Complementary recovery \& non-Clifford gates} \label{sect:complementary}

In the previous section we argued that HaPPY codes generally cannot admit transversal non-Clifford gates. In this section we turn our attention to wider families of holographic codes that have been introduced and studied more recently.
The natural question that occurs is whether the transversal Clifford restriction is a particular attribute of HaPPY codes, or a more general phenomenon due to some underlying physical reason that extends to these other families.
In \cref{sect:simple} we argue for the latter -- specifically that it is a consequence of any code for which regions satisfying complementary recovery are sufficiently common.
In fact, the argument we present can even apply if no regions satisfy complementary recovery, as long as some almost do (see \cref{sect:relaxGC}).
We extend to finite-size regions of the bulk in \cref{sect:regions}, codes with only approximate stabilizer encoding maps in \cref{sect:approx}, non-transversal locality-preserving gates in \cref{sect:nontransversal}, and codes with non-trivial centers in \cref{sect:center}.

\subsection{Complementary recovery \& transversal non-Clifford gates} \label{sect:simple}

First, we begin with a proof that satisfying complementary recovery for every region is sufficient to exclude transversal non-Cliffords.
A stronger variant of this theorem is one in which only every \emph{connected} region needs to satisfy complementary recovery.
This is stronger for two reasons: firstly, the assumption is less stringent and applies to more codes, and secondly, the assumption of $d_c\geq 2$ is weaker than $d\geq 2$.
\begin{theorem} \label{thm:complementary}
	Consider a holographic stabilizer code with encoding isometry $V$, and a particular choice of local bulk subsystem $i$.
	Suppose that every (connected) region satisfies complementary recovery, and the (connected) code distance is at least $d_{(c)}(\{i\}) \geq 2$.
	Then any transversal dressed-CSP unitary $U$ can only act on this subsystem via an element of the Clifford group, $VUV^{\dagger} = U_L \otimes U_J,\ U_L \in \mathcal{C}_2$.
\end{theorem}
\begin{proof}
	Let the region $R$ be the smallest non-correctable region.
	As $R$ is non-correctable, the entanglement wedge of its complement $\mathcal{E}[R^c]$ cannot contain $i$.
	By complementary recovery, this means that $i \in \mathcal{E}[R]$, and so $R^c$ must be correctable.
	Since we have a (connected) region of size $d_{(c)}(\left\{ i \right\})$ with a correctable complement, the price must be at most $d_{(c)}(\left\{ i \right\})$, and therefore $p_{(c)}(\left\{ i \right\}) = d_{(c)}(\left\{ i \right\})$ using \cref{lem:pricedistance}.
	Thus by assumption $2 \leq p_{(c)}(\left\{ i \right\}) \leq 2 d_{(c)}(\left\{ i \right\}) - 2$.
	By \cref{lem:PYprice}, any transversal dressed-CSP unitary can only implement an element of the Clifford group, $\mathcal{C}_2$.
\end{proof}

Of course, insisting on complementary recovery for every region, or even every connected region, is highly restrictive.
We do not even know if such a holographic stabilizer code exists\footnote{The random tensor networks of Ref.~\cite{haydenHolographicDualityRandom2016} provide non-stabilizer examples.}.
Fortunately, the following theorem requires a much weaker sufficient condition to exclude the possibility of transversal implementation of non-Clifford gates on a logical subsystem. 
The idea is that we need just one region $R$ satisfying complementary recovery to exist that would be suitable for use in the proof of \cref{thm:complementary}. This region $R$ need only satisfy a few basic properties for the proof to work.
\begin{theorem} \label{thm:lessComplementary}
	Consider a holographic stabilizer code with encoding isometry $V$, and a particular choice of local bulk subsystem $i$.
	Suppose there exists a non-correctable (connected) region $R$ with size $2\leq |R|\leq 2 d_{(c)}(\left\{ i \right\})-2$ that obeys complementary recovery\footnote{We use $d_{(c)}$ to denote $d$ if $R$ is not connected, and $d_c$ if it is.}.
	Then any transversal dressed-CSP unitary $U$ can only act on this subsystem via an element of the Clifford group, $VUV^{\dagger} = U_L \otimes U_J,\ U_L \in \mathcal{C}_2$.
\end{theorem}
\begin{proof}
    The proof proceeds exactly as that of \cref{thm:complementary}.
\end{proof}
\begin{figure}[t]
	\includegraphics[width=0.7\linewidth]{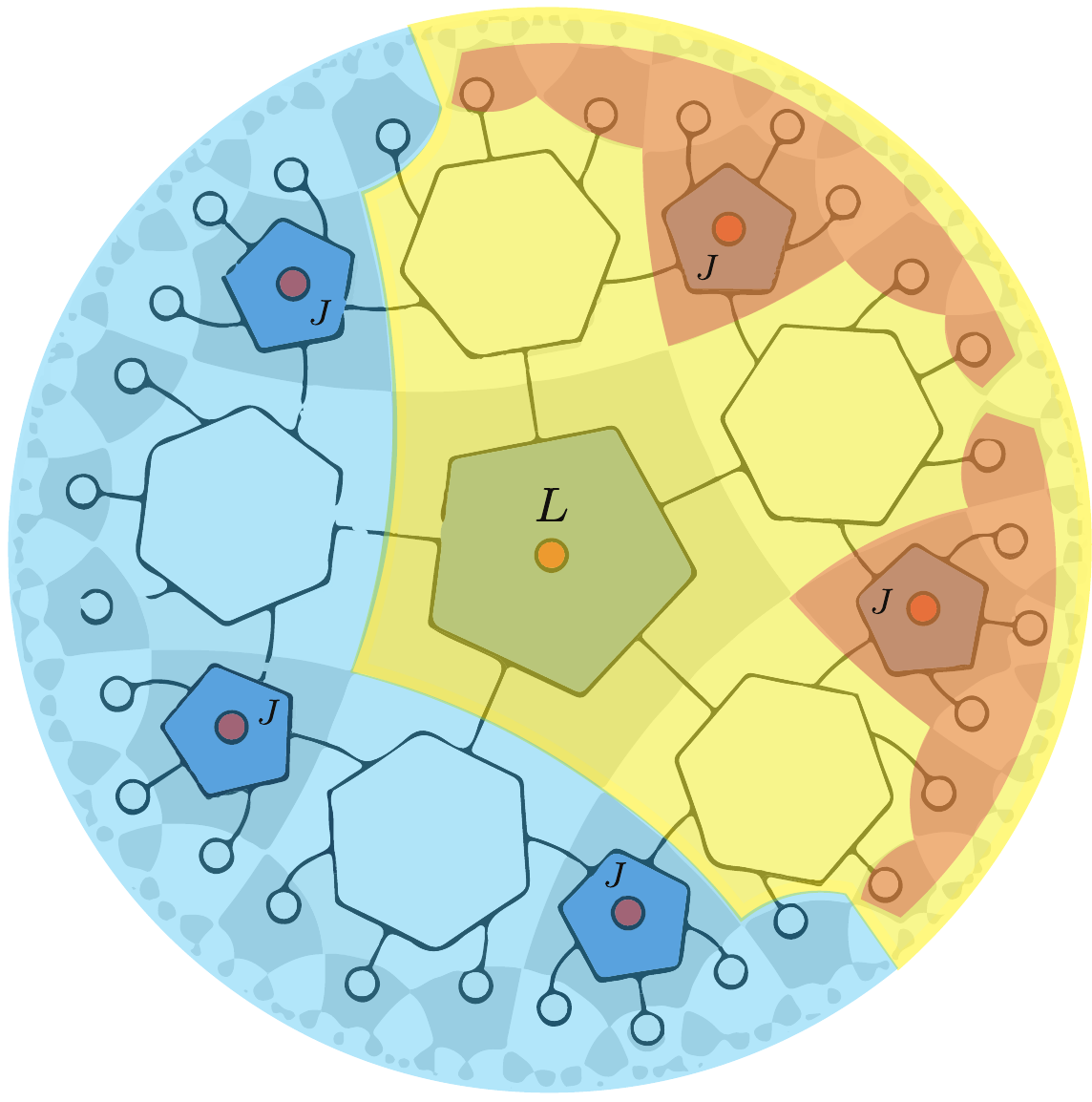}
	\caption{An alternating pentagon-hexagon HaPPY code. 
	We show how to construct three appropriate regions to be used in the proof of \cref{thm:lessComplementary}.
	In this case, the logical subsystem consists of only the central bulk qubit, and its connected code distance $d_c$ is $13$; an example of a connected non-correctable region $R$ with this size, along with its entanglement wedge, is highlighted in yellow.
	In this case, $R$ obeys complementary recovery, and hence can be used to prove Theorem~\ref{thm:lessComplementary} as follows: 
	due to complementary recovery, $R^c$ (highlighted along with its entanglement wedge in blue), is correctable.
	Furthermore, the two regions highlighted in red must be correctable as they are smaller than $R$, the smallest non-correctable region; thus the boundary partitions into three correctable regions. 
	\Cref{lem:PY} then implies that a non-Clifford gate cannot be applied transversally on the central qubit.
}
	\label{fig:alternating}
\end{figure}
An example of the construction used in this proof is shown in \cref{fig:alternating}.
In many holographic codes, such as certain instances of HaPPY, we expect the connected code distance to be significantly larger than the code distance, making the connected form of this theorem again more useful%
\footnote{Unlike \cref{sect:complementary}, the $d_c$ variant is not strictly stronger than the $d$ variant, for example a code may exist where only disconnected regions satisfy complementary recovery.}.

In any case, the definition of (connected) code distance implies the existence of at least one non-correctable region of size $d_{(c)}$; in addition to a number of (connected) supersets of this region that are also non-correctable and smaller than $2 d_{(c)}-2$.
As long as regions obeying complementary recovery are not \emph{too} rare, which they should not be for a purported holographic code, at least one such region will obey complementary recovery and the theorem will apply.
For example, although HaPPY codes do not generally obey complementary recovery, \cref{fig:alternating} shows an example of a HaPPY code in which an appropriate region can still be found to apply \cref{thm:lessComplementary}.

One can also construct a variant for any definition of size $\sigma(R)$, as per \cref{def:generalDistance}, so long as it is always possible to break a region into two smaller regions according to this notion of size.
As long as the assumptions are satisfied for any such notion of size, the conclusion holds.

One can strengthen this theorem by noting that the three regions do not necessarily need to be correctable.
Two regions need only be dressed-cleanable, which is a weaker condition in general -- relaxing the assumptions necessary to apply the theorem.
For example, the red and blue regions of the HaPPY code in \cref{fig:pauli} are both dressed-cleanable but not correctable.
We remark that this freedom is irrelevant in a code with all regions satisfying complementary recovery, due to \cref{lem:CR-DC} -- which also suggests holographic codes should not generally have too large a difference between dressed-cleanable and correctable regions.

It follows from the proof that if $R$ is close to $d_{(c)}$ in size, then there is significant ``room to spare'' when dividing the boundary system into three correctable regions.
That is, $R_1$ and $R_2$ are not only smaller than the smallest non-correctable region; they are roughly \emph{half} its size.
Instead of choosing $R_0$, $R_1$ and $R_2$ to be mutually non-intersecting, we could instead choose them to have some overlap while still being correctable and satisfying $R = R_0 \cup R_1\cup R_2 $.
Such an overlap proves useful in many of our subsequent extensions of this simple initial setup.

For example, it proves useful in generalizing beyond the case where  the logical subsystem $S\subset \left\{ 1,\dots ,n_L \right\}$ is chosen to consist of only a single local encoded subsystem $i$.
First, let us point out why the argument as presented above breaks down if multiple local encoded subsystems are included -- for example, suppose that two local encoded subsystems are contained in the logical subsystem, $S = \{i,j\}$.
Then it may be the case that $i\in \mathcal{E}[R^c]$ and $j \in \mathcal{E}[R]$, in which case both $R$ and $R^c$ are non-correctable with respect to $\left\{ i,j \right\}$ -- even if complementary recovery is satisfied.
By contrast, with a single local encoded subsystem, complementary recovery guaranteed that either $R$ or $R^c$ be correctable.
Nonetheless, in many cases three appropriate regions can still be constructed for a collection of sites in the bulk; an example in a HaPPY code is shown in \cref{fig:threesite}. 
We argue below that due to the existence of ``room to spare'' in the above argument, this should generally be the case.

\subsection{Finite-size bulk regions}\label{sect:regions}

We now generalize our argument to bulk regions with a finite, non-zero size. First we attempt to provide some intuition: 
We have seen that a logical subsystem associated to a point in a holographic code satisfies (connected) distance equals (connected) price. For extended bulk regions this equality no longer holds. However, from AdS/CFT we expect for the logical subsystem associated to a ball of radius $\delta$ in the bulk $p_c \leq d_c + \kappa(\delta)$, where $\kappa(\delta)$ grows exponentially with $\delta$ for any spatial dimension. Hence for sufficiently small bulk regions of a code with a holographic structure we expect $2\leq p_c \leq 2d_c -2$ is still satisfied. By Lemma~\ref{lem:PYprice} this implies that transversal gates can only act via Cliffords on the logical subsystem associated to any such bulk region. 
This argument is extended to locality-preserving gates for logical regions that are sufficiently deep within the bulk below. 

To make this generalization more rigorous we need to make reference to the geometric structure of the code, as introduced in \cref{sect:geometry}.
Specifically, we use the notion of a $\kappa_R$-smooth entanglement wedge map.
As discussed where it was introduced, $\kappa_R$ is typically exponentially large in $r$; however the amount of ``room to spare'' available from the proof of \cref{thm:lessComplementary} is proportional to the code distance $d$, which is also expected to be exponential in $r$ for deep bulk inputs.
Thus we are usually able to trade this overlap for the inclusion of a size $\mathcal{O}(1)$ region in the bulk.
{
\begin{theorem} \label{thm:kappa}
	Consider a holographic stabilizer code with encoding isometry $V$, and a particular choice of local bulk subsystem $S$ of diameter $D_S$ with a $\kappa_R$-smooth entanglement wedge. Suppose there exists a non-correctable connected region $R$ with size $2\le \abs{R} \le 2d_{c}(S)-2-\kappa_R(D_S)$ that obeys complementary recovery. Then any transversal dressed-CSP unitary $U$ is constrained to act on this subsystem via an element of the Clifford group, $VUV^{\dagger} = U_L \otimes U_J,\ U_L \in \mathcal{C}_2$.
\end{theorem}

\begin{proof}
	As $R$ is not correctable and obeys complementary recovery, there is at least one subsystem $i\subset S$ contained in $\mathcal{E}(R)$.
	$S \subset B(\mathcal{E}(R),D_S)$ because the distance between two points in $S$ is at most $D_S$.
	Thus, if $x=\kappa_R(D_S)$ then there exists a connected boundary region $R_x$ such that $S \subset \mathcal{E}(R_x)$, meaning its complement $R_x^c$ is correctable.
	Thus the code price is upper bounded by $|R_x| = |R|+x$, which in turn is upper bounded by $2d_{c}(S)-2$, so we can apply \cref{lem:PYprice} to conclude the proof.
    \end{proof}
    
}

\begin{figure}[t]
	\centering
	\includegraphics[width=0.7\linewidth]{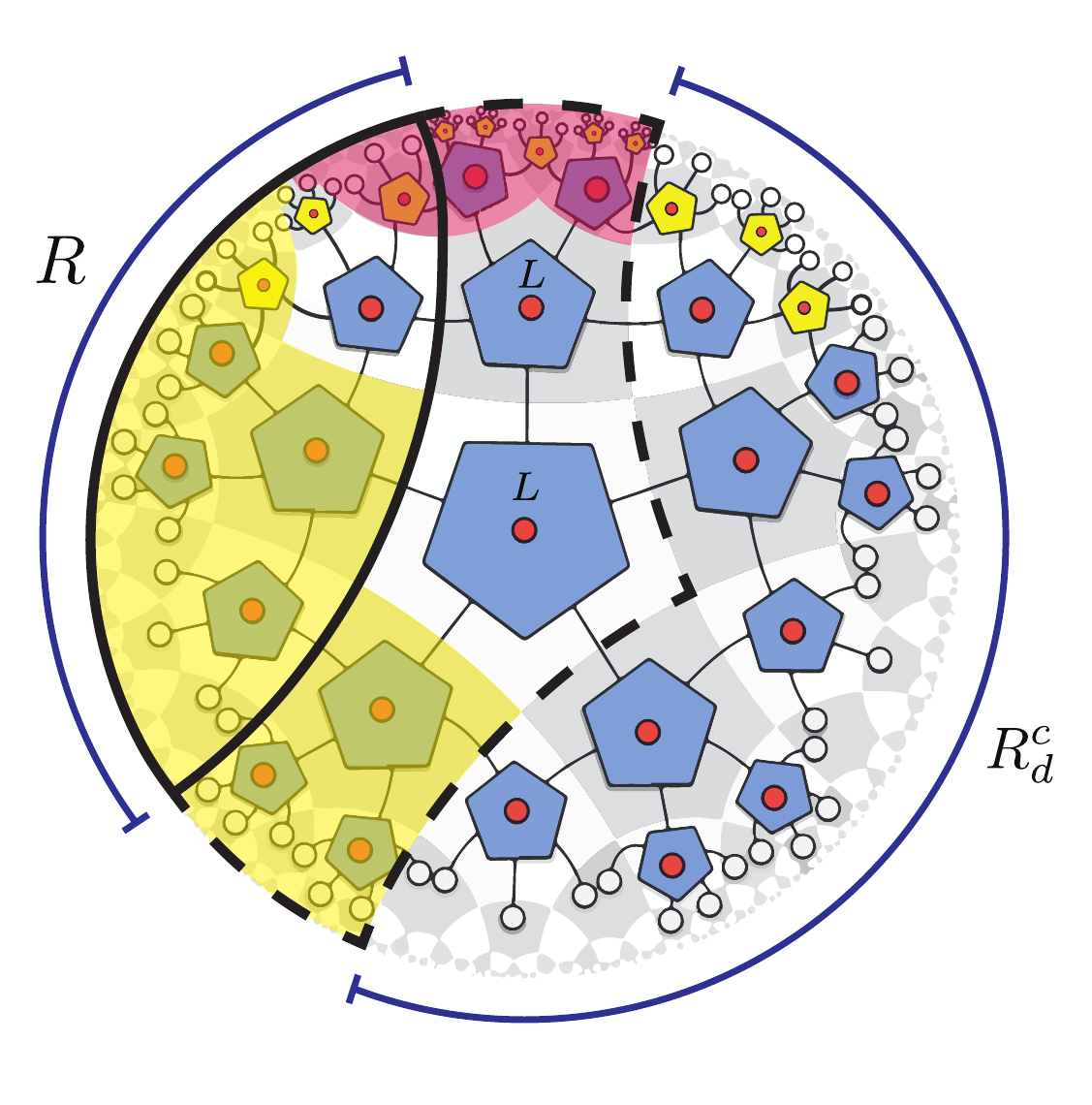}
	\caption{An example with a bulk region containing more than one site, illustrating theorem \ref{thm:kappa}. More perfect tensors than usual, colored (not highlighted) yellow, are included in order to be able to pick boundary regions with greater precision. A non-correctable region $R$ has an entanglement wedge bordering the logical region. $R$ needs to be expanded only slightly, to $R_d$, in order to include the whole of the logical region. Thus $R_d^c$ is correctable. Furthermore, we can break $R_d$ into two correctable regions, highlighted in yellow and red. Thus we have found three correctable regions and \cref{lem:PY} can be applied.}
	\label{fig:threesite}
\end{figure}

\subsection{Almost complementary recovery} \label{sect:relaxGC}

We now move on to consider a more general case where the entanglement wedge map is smooth, but the region $R$ only almost obeys complementary recovery. 
In this case, $i$ may not belong to the entanglement wedge $\mathcal{E}[R]$, but a slightly 
larger $R^+\supset R$ can be chosen such that $i\in \mathcal{E}[R^+]$, and thus $R^{+c}$ is correctable with respect to~${\left\{ i \right\}}$.

\begin{definition} \label{def:approxCR}
	We say that a region $R$ satisfies $\varepsilon$-almost complementary recovery if it can be expanded by $\varepsilon $ in all directions, such that the new entanglement wedge includes the complement of the entanglement wedge of $R^c$, i.e.\ $\mathcal{E}[R^c]^c \subset  \mathcal{E}[B(R,\varepsilon ) ] $%
	\footnote{
	One can also express this in terms of the bulk rather than the boundary, i.e.\ by requiring $\mathcal{E}[R^c]^c \subset B(\mathcal{E}[R],\delta ) $, and then enforcing $\kappa_R(\delta) \leq \varepsilon$. 
	This has the nice physical interpretation that $\delta$ bounds the thickness of the residual region contained in neither $\mathcal{E}[R]$ nor $\mathcal{E}[R^c]$ -- for example, connected regions of HaPPY codes satisfy it with $\delta = 1$ -- however formulating it in terms of the boundary makes \cref{thm:approxCR} slightly stronger.}
	.
\end{definition}

\begin{theorem} \label{thm:approxCR}
	Consider a holographic stabilizer code with encoding isometry $V$, and a particular choice of local bulk subsystem $S$ of diameter $D_S$ with a $\kappa_R$-smooth entanglement wedge. Suppose there exists a non-correctable connected region $R$ with size ${2\le \abs{R} \le 2d_{}(S)-2- \kappa_R(D_S) -  \varepsilon }$ that obeys $\varepsilon $-complementary recovery. Then any transversal dressed-CSP unitary $U$ can only act on this subsystem via an element of the Clifford group, $VUV^{\dagger} = U_L \otimes U_J,\ U_L \in \mathcal{C}_2$.
\end{theorem}

\begin{proof}
    The proof is a straightforward combination of \cref{def:approxCR} with the proof of \cref{thm:kappa}.
\end{proof}

\subsection{Approximate stabilizer codes} \label{sect:approx}
AdS/CFT is believed to not be an exact quantum error-correcting code, as its error-correction properties only hold to first order in a perturbative expansion~\cite{almheiriBulkLocalityQuantum2015}.
Instead, it is best understood as an approximate code~\cite{almheiriBulkLocalityQuantum2015}.
One family of toy models~\cite{caoApproximateBaconShorCode2020} attempt to capture this aspect of AdS/CFT in terms of an encoding map that is only \emph{approximately} equal to an encoding isometry of a stabilizer code.
Specifically, the map $V:\mathcal{H}_L \to \mathcal{H}_P$ is given by
\begin{align}
	V = V_0 + \delta V,
	\label{eq:approximateStabilizer}
\end{align}
with $V_0$ an encoding isometry for a stabilizer code, and $\delta V$ some small perturbation; $\|\delta  V\|_1 = \varepsilon $ for some small $\varepsilon >0$.

In such a case, it is straightforward to see that a CSP unitary must implement a logical gate which is approximately equal to an element of the Clifford group.
To see this, note that $V_0$ is an exact holographic stabilizer code, and thus the unitary applied to this code implements some element of the Clifford group, $V_0^{\dagger} U V_0= U_L \in \mathcal{C}_2$.
Then one can show using H\"older's inequality that 
\begin{align}
	\| V^{\dagger}UV - U_L \|_1 &=  \| (V_0 + \delta V)^{\dagger} U \left( V_0 + \delta V \right) - V_0^{\dagger} U V_0 \|_1 \\
	&\leq  \| V_0 U \delta V^{\dagger}\|_1  + \| \delta V U  V_0^{\dagger} \|_1 +\|  \delta V^{\dagger} U \delta V\|_1 \\
	&\leq  \| \delta V\|_1  + \| \delta V\|_1  + \| \delta V \|_1  \| \delta V^{\dagger}\|_{\infty } \\
	&\leq 2 \varepsilon +\varepsilon ^2
	\label{eq:approximateClifford}
\end{align}
i.e.\ the logical operation implemented by the full code $V$ is close to an element of the Clifford group.

\subsection{Non-transversal locality-preserving gates} \label{sect:nontransversal}
We now combine the Pastawski-Yoshida lemma \cref{lem:PY} with the geometric structure introduced in \cref{sect:geometry} to extend the previous result to the case of unitaries that are not necessarily transversal, but are still locality-preserving -- i.e.\ that have some finite spread $s_U>0$.

This generalization is simplest to formulate when the spatial boundary geometry is one-dimensional.
\begin{theorem}\label{thm:1d}
    For a subsystem stabilizer code $S$ with a one-dimensional boundary, if there is a non-correctable region $R$ of size $d_c(\{i\})$ that satisfies complementary recovery then any dressed-CSP unitary operator $U$ with $s_U<\frac{d_c}{6}$ implements a logical unitary in $\mathcal{C}_2$ on the logical system $\{i\}$.
\end{theorem}
\begin{proof}
 To apply \cref{lem:PY}, we need to partition the boundary space into three regions such that one is correctable, one is correctable even after expansion by a distance of $s_U$ either side, and one is correctable after expansion by $2s_U$ either side.
 Furthermore, the definition of $d_c$ means that any connected region smaller than $d_c$ is correctable.
 Thus we can proceed as follows.
 
 Parametrize the interval $R$ as $[0,d_c]$, and choose $R_1$ as $[0,\frac{2d_c}{3}]$ and $R_2$ as $[\frac{2d_c}{3},d_c]$.
Then for any unitary with spread $s_U < \frac{d_c}{6}$, the extended regions $R_1^+ = B(R_1,s_U) = [-s_U,\frac{2d_c}{3}+s_U]$ and $R_2^+ = B(R_2,2s_U) = [\frac{2d_c}{3}-2s_U,d_c+2s_U]$ are both less than $d_c$ in length, and therefore remain correctable; $R_0$ also remains correctable, and thus \cref{lem:PY} applies. 
\end{proof}

More generally, this argument extends to higher dimensions as follows: let $d$ be the size of the smallest non-correctable region. Let $R$ be a connected region of size $d$. We then split $R$ into two regions $R_1$ and $R_2$ of sizes $\frac{2d}{3}$ and $\frac{d}{3}$, respectively. We assume that we can choose sufficiently nice regions $R_1$ and $R_2$ in the sense that the size of their extensions $B(R_i,s)$ do not increase too fast\footnote{
We expect $B(R_i,s)$ to scale as $|\partial R_i|s$ for small $s$, so a region with small boundary/bulk ratio, such as a ball in flat space, suffices.} with $s$. Then for sufficiently small $s_U$ the regions $R^{+}_{i} = B(R_i,2^{2-i}s_U)$ have size smaller than $d$ and thus remain correctable. We can then apply \cref{lem:PY} to achieve the desired result. 
Following this line of reasoning, \cref{thm:1d} generalizes straightforwardly to a higher dimensional setting where the boundary is a sphere and the region $R$ is a ball. %
At this time we do not know of a clean statement of the most general higher dimensional version of our result. 
This is in part due to the richer set of possibilities for the topology of connected regions in dimensions higher than one. 
Furthermore, 
degenerate cases exist where the region $R$ is itself a (thickened) boundary and as such the scaling of  $B(R_i,s)$ with $s$ for any subregion is neccesarily too fast to accommodate any positive spreading $s_U>0$.
However, recall that the region $R$ was chosen as a non-correctable region with minimal size.
In AdS/CFT, such regions are generally expected to be solid spheres \cite{Hubeny2012}, and as such the analogous regions in discrete models of AdS/CFT are expected to be approximately spherical and thus avoid the troublesome degenerate cases mentioned above.

\subsection{Entanglement wedge surface algebras with non-trivial centers} \label{sect:center}

Subsystem codes are a special case of \emph{operator-algebraic quantum error correction}, in which some general algebra of operators (the ``logical algebra'') is protected against error.
In subsystem codes, this is the algebra of bare logical operators $\{A_L \otimes \Id_J\}$.
An algebra of this form is special in that its only central elements are those proportional to the identity
.
More generally, the logical algebra may have a \emph{non-trivial center}, i.e.\ it may contain elements which commute with all other elements, but are not proportional to the identity on $L$.
The physical consequence of this is that the logical region is no longer associated with a Hilbert space  -- it is now identified with a logical \emph{algebra} instead~\cite{Beny2007,Beny2007b}.
Because the entanglement wedge map associates a bulk region with each boundary region, this means it now identifies each boundary region with an algebra and not a subsystem~\cite{almheiriBulkLocalityQuantum2015}.
As a result, the limit from \cref{eq:ewlimit} no longer applies, and the regions $\mathcal{E}[R]$ and $\mathcal{E}[R^c]$ can have a non-empty overlap $\mathcal{E}(R)\cap\mathcal{E}(R^c)\neq \varnothing$, which we refer to as the \emph{entangling surface}.
Because this algebra can be reconstructed on independent regions of the boundary, it must be associated with an abelian algebra, which also forms the center of the algebras associated with $\mathcal{E}[R]$ and $\mathcal{E}[R^c]$ respectively.

For example, in Ref.~\cite{donnellyLivingEdgeToy2017} codes were introduced (which we refer to as LOTE codes after the title of that paper) that are constructed from HaPPY codes, but with additional qubits associated to the entangling surface of any given region $R$.
However, only the algebra of operators generated by $\left\{ \Id, \sigma _Z \right\}$ acting on such a qubit can be recovered on either $R$ or $R^c$.
The full set of Paulis $\left\{ \Id, \sigma _X,\sigma _Y,\sigma _Z \right\}$ can still be recovered on the boundary, but not on $R$ or $R^c$ alone.
This means that with respect to this partition of the boundary, the data on the entangling surface is composed of classical bits.
This is meant to model an intuitive feature of AdS/CFT: the area of the entangling surface is information which ought to be recoverable from both $R$ and $R^c$. 

Even for LOTE codes we are still able to apply our methods to exclude transversal implementation of non-Clifford gates.
We do this fairly straightforwardly -- we simply ignore the additional algebraic structure, and proceed as if it were a subsystem code with local encoded subsystems as in \cref{eq:TPS}.
A qubit on the entangling surface has a subalgebra reconstructable on either $R$ or $R^c$, but its entire algebra can be reconstructed on neither -- thus it does not belong to the subsystem associated with either region's maximal entanglement wedge.
In other words, such qubits contribute to any given region's failure to obey complementary recovery.
Thus there are no non-trivial regions satisfying complementary recovery (as we have defined it here) in such a code.
Nonetheless, so long as regions satisfying complementary recovery are not too rare in the underlying HaPPY code, regions that \emph{almost} satisfy complementary recovery are not too rare in the corresponding LOTE code.
Thus we can apply the results of \cref{sect:relaxGC} to show that non-Cliffords cannot be implemented in the \emph{interior} of a region's entanglement wedge.
It is unclear how to define a Clifford operator on an algebra with a mixture of bits and qubits, so this appears to be the most that can be said -- locality-preserving gates cannot implement logical non-Cliffords on a slightly smaller quantum subalgebra that excludes the classical parts.

\subsection{Other levels of the Clifford hierarchy}\label{sect:levels}
For the special case that $U$ is a bare-CSP operator, \cref{lem:PYbare} implies a stronger result; namely that one can allow $R_0$ to only be dressed-cleanable rather than correctable. 
This can be used to strengthen the result further. 
In fact, in some cases it can even be used to show that only Pauli gates can be implemented, such as in \cref{fig:pauli}, by dividing the boundary into just two dressed-cleanable regions $R_0$ and $R_1$. 
This case of the bound is relevant if one wishes to implement a desired bare logical gate on a small bulk subsystem without causing any disturbances to other regions of the bulk. 

We remark that it is not possible to restrict a \emph{dressed-CSP} operator to implement something only from $\mathcal{C}_1$ in this way.
To do so would require expressing the boundary as a union of a dressed-cleanable and a correctable region, which is impossible; any bare-logical operator could be cleaned to $R_1$ using correctability of $R_0$, but \cref{lem:baretrivial} applied to $R_1$ means that the operator must be trivial -- a contradiction.

Finally, for sufficiently large bulk regions, the arguments of \cref{sect:nontransversal} do not apply.
The reason for this is that all correctable regions are too small to cover the entire boundary using only three of them; some greater number would be required.
In these cases, a restriction on locality-preserving gates still exists, depending on the number $k$ of required correctable regions.
Thus, even if a non-Clifford cannot be ruled out for very large bulk regions, some restriction still exists corresponding to a higher level of the Clifford hierarchy, $\mathcal{C}_{k-1}$.

\begin{figure}[t]
	\centering
	\includegraphics[width=0.7\linewidth]{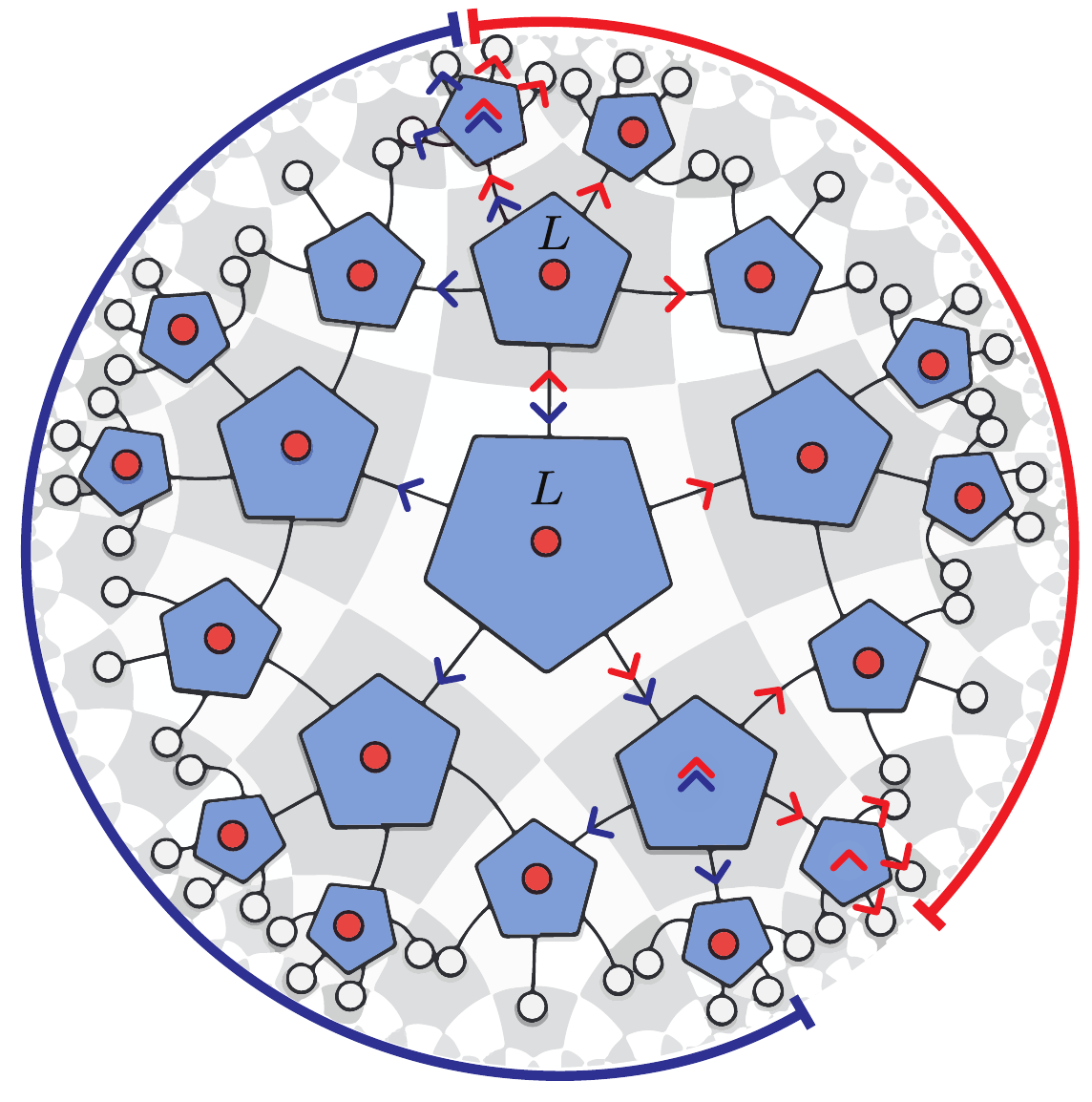}
	\caption{In this case, with a two-site logical region, we can make an even stronger claim about what can be implemented by a \emph{bare}-CSP operator using \cref{lem:PYbare}.
	The complements of both the red and blue solid regions are dressed-cleanable, as illustrated by the arrow assignments.
Thus they can be used in \cref{lem:PYbare} with $m=1$, meaning that a bare-CSP operator can only apply a Pauli operator to the central qubits; even a non-Pauli Clifford is forbidden. }
	\label{fig:pauli}
\end{figure}

\begin{theorem}\label{thm:HigherCliffords}
    For a subsystem stabilizer code $S$ with a one-dimensional boundary, if there is a region $R$ of size $d_c(S)$ that satisfies complementary recovery then any dressed-CSP unitary operator $U$ with $s_U<\frac{d_c}{2(2^m-1)}$ implements a logical unitary in $\mathcal{C}_{m}$ for $m\ge 2$.
\end{theorem}
\begin{proof}
    As in Sec.~\ref{sect:nontransversal}, we subdivide the region $R$ into $m$ consecutive regions $R_{j=1,\dots,m}$ of sizes $\frac{2^{m-j}d_c}{(2^m-1)}$. Then $R_{j}^{+} = B(R_j,2^{j}s_U)$ has size $\frac{2^{m-j}d_c}{(2^m-1)}+2^{j+1}s_U < \frac{2^{m-j}+2^{j+1}}{2^{m}-1}d_c \le d_c$. We can thus apply \cref{lem:PY} and obtain the result.
\end{proof}

The above argument can be generalized to higher dimensions along similar lines to the discussion at the end of \cref{sect:nontransversal}. 
Again we assume that there is a boundary region $R$ of the smallest non-correctable region size that can be decomposed into a set of subregions $R_i$,  $i=1,\dots,m$, that are sufficiently nice in the sense that the size of their extensions $B(R_i,s)$ do not grow too fast with $s$. 
Then for sufficiently small $s_U$, the argument presented above for \cref{thm:HigherCliffords} can be applied straightforwardly, leading to the same conclusion. 
Similar to \cref{sect:nontransversal}, we do not currently have a clean statement of the most general version of this result but we expect the statement for a solid spherical boundary region to be the most relevant.

\section{Discussion \& conclusions}

In this work we have established a specialisation and strengthening of the Eastin-Knill theorem for holographic subsystem stabilizer codes, inspired by analogous work due to Bravyi-Koenig and Pastawski-Yoshida for topological stabilizer codes. 
Specifically, we have shown that sufficiently locality-preserving operations on the physical qubits of a holographic stabilizer code can implement only logical Clifford gates on sufficiently small regions of the bulk. 
This result was extended in several ways, including to approximate stabilizer encodings and to codes in which algebras with non-trivial centers live on the surfaces of entanglement wedges.
We have further shown that upon weakening our assumptions, locality-preserving gates are still restricted to implement an element from a fixed level of the Clifford hierarchy. 
In the course of deriving our results we also introduced a general definition of the maximal entanglement wedge for arbitrary subsystem codes which may be of independent interest.  
While we have focused on the simple case of the hyperbolic disc in our figures, we remark that our results also apply to non-trivial bulk topologies and higher dimensional spaces.

There is a recurring theme in fault-tolerant quantum computation that the set of fault-tolerantly implementable gates is more severely restricted the better the error-correction properties of a code are.
This makes some intuitive sense -- better protected logical information should also be harder to manipulate.
For example, Eastin-Knill's theorem establishing the impossibility of universal transversal gate sets only applies when the code distance is greater one; the Brayi-Koenig theorem and Pastawski-Yoshida's subsystem generalization require a macroscopically large code distance; and another result of Pastawski-Yoshida restricts transversally implementable gates more strongly the higher the loss threshold is.
Similarly, our results require assumptions about the quality of the code; namely that the distance of a region is ``large enough'' with respect to the price of the region.

An interesting future direction is the extension of our results restricting locality-preserving gates to the setting of the full AdS/CFT duality, possibly by exploiting conformal invariance. This approach is inspired by the generalization of the Bravyi-Koenig bound to all (2+1)D TQFTs (on closed manifolds)~\cite{beverlandProtectedGatesTopological2016} that was established by exploiting consequences of topological invariance. 
As explained in the introduction, it has been successfully argued that AdS/CFT supports no global symmetries~\cite{Harlow2019}. 
While this result partially overlaps with our goal, our aim is somewhat stronger, restricting the possible bulk evolutions generated by arbitrary locality- and codespace-preserving evolutions on the boundary with no requirement that they take the form of a global symmetry. 
This could have potentially interesting implications for the prospect of implementing bulk Hamiltonian evolution via locality-preserving operations on the boundary.

There is an apparent tension between the results of this paper and the intuition that there ought to exist states in AdS/CFT which contain a quantum computer in the bulk. In a short amount of time, this quantum computer ought to be able to perform, at least approximately, an arbitrary element of $SU(2^k)$ on $k$ qubits for $k\geq 2$ (but not much greater). But short time evolution is implemented on the boundary by a locality preserving unitary, which we have seen can only implement Cliffords in holographic stabilizer codes. Of course, in AdS/CFT the local bulk or boundary degrees of freedom are larger than simply qubits, so a fairer comparison should use qudit rather than qubit stabilizer codes. However, even for large qudit dimension it is impossible for the Clifford group to contain a subgroup arbitrarily approximating $SU(2^k)$ \cite{turing}. Thus it appears unlikely that exact or approximate (as defined in \cref{sect:approx}) stabilizer codes alone can capture this aspect of AdS/CFT. 
Our results can hence be interpreted as placing severe limitations on how well stabilizer toy models are able to capture aspects of holography to do with locality-preserving boundary evolutions. 
In a future work we plan to show that by adding additional constraints to the logical bulk of a qudit stabilizer code, it is indeed possible to build an approximate holographic code which can transversally implement any element of $SU(2^k)$.

It would be very interesting to combine our restrictions on locality-preserving operations based on holographic structure and locality with the approximate covariant code trade-off bounds derived in Ref.~\cite{Faist_ContSymApproxQEC2020}. 
We expect that both of these directions are relevant to studying the action of locality-preserving evolutions in the full AdS/CFT duality which is expected to lead to approximate quantum codes with a holographic locality structure. 

\acknowledgments

We thank Patrick Hayden for helpful discussions and comments on the manuscript.
We also thank Jon Sorce for pointing us to Ref.~\cite{Hubeny2012}.
KD is supported by the Center for Science of Information (CSoI), an NSF Science and Technology Center, under grant agreement CCF-0939370. 
SC acknowledges support from the Knight-Hennessy Scholars program.
DW acknowledges support from the Simons Foundation.

\newpage

\bibliographystyle{apsrev4-1}
\bibliography{samlibrary}

\appendix

\section{Correctability properties for regions of subsystem codes} \label{app:regions}

In \cref{sect:background}, we introduced two correctability properties that a region of a subsystem code can have -- correctability and dressed-cleanability.
In this appendix, we explore a broader family of such properties, and present the known relationships between them. 

\subsection{Subsystem codes}
To begin with, in general subsystem codes, a notion of dressed-cleanable can be defined analogous to \cref{def:dressedCleanable} but for any dressed operator, not just Pauli operators.
\begin{definition}\label{def:subsystem}
A region $R$ of a subsystem code is defined to be
	\begin{itemize}
		\item \textbf{\emph{bare-cleanable}} (aka ``correctable'') iff for any logical operator $A_L$ there exists a bare-CSP operator $A$ supported on $R^c$ which implements $A_L \otimes \Id_J$.
	\item \textbf{\emph{dressed-cleanable}} iff for any logical operator $A_L$ there exists a dressed-CSP operator $A$ supported on $R^c$ which implements $A_L \otimes A_J$.
	\item \textbf{\emph{bare-trivial}} iff any bare-CSP operator supported on $R$ implements a logical operation of the form $c\Id_L \otimes \Id_J$, with $c\in \mathbb{C}$.
	\item \textbf{\emph{dressed-trivial}} iff any dressed-CSP operator supported on $R$ implements a logical operation of the form $\Id_L \otimes A_J$, for some $A_J \in \mathcal{B}(\mathcal{H}_J)$.
	\end{itemize}
\end{definition}
We remark that \cref{lem:dressedtrivial,lem:baretrivial} can then be expressed as meaning that bare-cleanable implies dressed-trivial and dressed-cleanable implies bare-trivial.
Furthermore, \cref{lem:BCDC} means that bare-cleanable implies dressed-cleanable, and it is easy to show similarly that dressed-trivial implies bare-trivial.
We can summarise these results, along with \cref{lem:CR-DC} for regions satisfying complementary recovery, in \cref{fig:subsystem}.

\newlength{\noddist}
\setlength{\noddist}{4cm}
\newlength{\hordist}
\setlength{\hordist}{4cm}
\begin{figure}[htp]
\makebox[\linewidth]{
\begin{tikzpicture}[node distance = \noddist, scale=0.85, every node/.style={transform shape}]
	\node[align=center](BC){\textbf{Bare-cleanable}\\
	\textbf{(Correctable)}};
	\node[right of=BC,node distance=\hordist,align=center](DC) {\textbf{Dressed-cleanable}};
	\node [below of=BC,node distance=\noddist] (DT) {\textbf{Dressed-trivial}};
	\node [right of=DT,node distance=\hordist] (BT) {\textbf{Bare-trivial}};
    % Draw edges
	\draw[thick,double,->] (BC) to node[black,right,align=center]{}(DT);
	\draw[thick,double,->] (DC) to node[black,right,align=center]{}(BT);
	\draw[thick,double,->] (BC) -- node[black,right,align=center]{}(DC);
	\draw[thick,double,->] (DT) -- node[black,right,align=center]{}(BT);
	\draw[blue,thick,double,->] (BT) -- node[blue,sloped,above,align=center]{Complementary\\Recovery}(BC);
\end{tikzpicture}}
\caption{\textbf{Relationships between properties of physical regions in subsystem codes.} 
The blue relationship holds for regions with complementary recovery, such as those in holographic codes.
}
\label{fig:subsystem}
\end{figure}

\subsection{Stabilizer subsystem codes}
In the special case of a \emph{stabilizer} subsystem code, there is some additional structure.
The definition in \cref{def:subsystem} are still valid, but we could instead choose to define analogous properties that only require the respective properties to hold for Pauli operators.
It turns out that most of these properties are equivalent to those in \cref{def:subsystem}, for example:
\begin{lemma}\label{lem:paulibare}
	A region $R$ of a stabilizer subsystem code is bare-cleanable if and only if for any logical Pauli operator $P_L$, there exists a bare-CSP Pauli operator $P$ supported on $R^c$ which implements $P_L \otimes \Id_J$.
\end{lemma}
\begin{proof}
	$(\implies)$
	For any logical operator $P_L$, bare-cleanability tells us there is a bare-CSP operator $A$ supported on $R^c$ that implements $P_L \otimes \Id$.
	$A$ can be decomposed into Pauli operators supported on $R^c$, as $A = \sum_{i}^{}\lambda _i P^i$.
	Since $A$ commutes with every stabilizer, each of these Pauli operators $P^i$ must also commute with all the stabilizers, and therefore be codespace-preserving.
	Furthermore, the basis of the logical space has been chosen such that CSP physical Pauli operators always implement logical Paulis, so each $P^i$ is a dressed-CSP operator supported on $R^c$ implementing $V P_i V^{\dagger} = P_L^i \otimes P_J^i$.
	Since $A$ implements $P_L \otimes \Id$, we have $VAV^{\dagger} = \sum_{i}^{}\lambda _i P_L^i \otimes P_J^i = P_L \otimes \Id_J $.
	Since the Paulis form a basis of logical operators, at least one $P^i$ is a bare-CSP implementing $P_L \otimes \Id _J$, completing the proof.

	$(\impliedby) $ Any logical operator $A_L$ can be decomposed into Paulis as $\sum_{i}^{} \lambda _i P_L^i$.
	Then each of these Paulis can be implemented as a bare-CSP operator $P^i$ on $R^c$.
	Then $\sum_{i}^{}\lambda _i P^i$ is a bare-CSP operator supported on $R^c$ implementing $A_L$, concluding the proof.
\end{proof}
Similar proofs can be constructed for the bare-trivial and dressed-trivial properties; however the reasoning breaks down for the dressed-cleanable property.
We remark that in the converse direction of the above proof, it was crucial that the linear combination of bare-CSP operators $\sum_{i}^{} \lambda _i P^i$ was also a bare-CSP operator.
However, this does not apply to dressed-CSP operators, which are not closed under linear combination.
Thus, we need to define two separate notions of dressed-cleanability for stabilizer subsystem codes -- one based on general operators and one based on Paulis.
The latter of these two is the one that was used in the main text and referred to simply as ``dressed-cleanable'' -- because it is the one which is useful for the Pastawski-Yoshida lemma, \cref{lem:PY}.
\begin{definition}\label{def:stabsubsystem}
A region $R$ of a subsystem stabilizer code is defined to be
	\begin{itemize}
	\item \textbf{\emph{general dressed-cleanable}} iff for any logical operator $A_L$ there exists a dressed-CSP operator $A$ supported on $R^c$ which implements $A_L \otimes A_J$.
	\item \textbf{\emph{Pauli dressed-cleanable}} iff for any logical Pauli operator $P_L$ there exists a dressed-CSP operator $P$ supported on $R^c$ which implements $P_L \otimes P_J$.
	\end{itemize}
\end{definition}
It can still be shown analogously to the first part of the proof of \cref{lem:paulibare} that general dressed-cleanable regions are also Pauli dressed-cleanable.

The one other difference for stabilizer subsystem codes is that the converse directions for \cref{lem:baretrivial,lem:dressedtrivial} are also true (see Ref.~\cite{bravyi_SubsystemLocal}), as mentioned in the main text.
	These are shown specifically for Pauli dressed-cleanability, and not general dressed-cleanability.
	Thus the web of relationships for stabilizer subsystem codes is shown in \cref{fig:stabilizersubsystem}.

\begin{figure}[htp]
\makebox[\linewidth]{
\begin{tikzpicture}[node distance = \noddist, scale=0.85, every node/.style={transform shape}]
	\node[align=center](BC){\textbf{Bare-cleanable}\\
	\textbf{(Correctable)}};
	\node[right of=BC,node distance=\hordist,align=center](DC) {\textbf{General}\\\textbf{Dressed-cleanable}};
	\node[right of=DC,node distance=\hordist,align=center](PDC) {\textbf{Pauli}\\\textbf{Dressed-cleanable}};
	\node [below of=BC,node distance=\noddist] (DT) {\textbf{Dressed-trivial}};
	\node [right of=DT,node distance=2\hordist] (BT) {\textbf{Bare-trivial}};
    % Draw edges
	\draw[thick,double,<->] (BC) to %[out=300,in=60]
	node[black,right,align=center]{}(DT);
	%\draw[thick,red,double,->] (DT) to [out=120,in=240] node[red,left,align=center]{Stabilizer\\codes only}(BC);
	\draw[thick,double,<->] (PDC) to %[out=240,in=120]
	node[black,right,align=center]{}(BT);
	%\draw[thick,red,double,->] (BT) to [out=60,in=300] node[red,right,align=center]{Stabilizer\\codes only}(PDC);
	\draw[thick,double,->] (BC) -- node[black,right,align=center]{}(DC);
	\draw[thick,double,->] (DC) -- node[black,right,align=center]{}(PDC);
	\draw[thick,double,->] (DT) -- node[black,right,align=center]{}(BT);
	\draw[blue,thick,double,->] (BT) -- node[blue,sloped,above,align=center]{Complementary\\Recovery}(BC);
\end{tikzpicture}}
\caption{\textbf{Relationships between properties of physical regions in subsystem codes.} 
The blue relationship holds for regions with complementary recovery, such as those in holographic codes.
}
\label{fig:stabilizersubsystem}
\end{figure}

To the best of our knowledge, this is the extent of the known relationships between different correctability properties in both subsystem and stabilizer subsystem codes.
This leaves many open questions: are dressed-trivial and bare-cleanable equivalent for any subsystem code?
Is the general dressed-cleanable property of stabilizer subsystem codes equivalent to Pauli dressed-cleanability after all, is it equivalent to bare-cleanability, or is it distinct from both?
In HaPPY codes, such as in \cref{fig:pauli}, regions are found which are Pauli dressed-cleanable and not bare-cleanable, which already shows that some of the above properties must be distinct.
Perhaps similar counter-examples can be found demonstrating that the relationships shown here are all that can be generally expressed without additional information.
We leave this to future work.

\newpage

\section{Homogen{e}ous planar hyperbolic HaPPY codes} \label{app:schlafli}

In Ref.~\cite{jahnCentralChargesAperiodic2020}, a range of possible HaPPY codes were considered by systematically exploring possible homogeneous hyperbolic tilings characterized by Schlafli symbols.
Schlafli symbols are a characterization of a regular polygonal tiling via a pair of integers ${n,k}$, such that $k$ $n$-sided polygons ($n$-gons) meeting at each vertex.
Such a tiling is hyperbolic if the total internal angles of each polygon exceeds $(n-2)\pi $, which is equivalent to
\begin{align}
	\frac{1}{n} + \frac{1}{k} \leq \frac{1}{2}.\label{eq:hyperbolic}
\end{align}
In Ref.~\cite{jahnCentralChargesAperiodic2020}, the perfect tensors of HaPPY are associated to polygons of this tiling; however for our purposes we consider the dual identification of these tensors with the \emph{vertices} of the graph instead.
Thus for a network with a rank $2N$ perfect tensor, we should have $k=2N-1$ as the number of in-plane legs.
These two conventions are related by Poincar\'e duality, where faces and vertices are interchanged, which corresponds to switching $n\leftrightarrow k$ in the Schlafli symbol.
The symmetry of \cref{eq:hyperbolic} shows that a tiling is hyperbolic if and only if its dual is hyperbolic, as expected.

Of course, the HaPPY tensor network is more than just a hyperbolic tiling filled with a network of perfect tensors; it must also be an isometry from the bulk to the boundary.
The former does not imply the latter, as one can see from the case where $k=3$.
In this case, a local bulk operator cannot be pushed out to the boundary; whenever an operator is pushed to a new vertex, there are two inputs (one from the leg ``pushed'' from the previous vertex, and one from the bulk input), so the result of the pushing is that some operator is pushed to both remaining legs of the tensor.
In this way, there is no control over the pushing process to guide the output towards the boundary; it inevitably loops back on itself as soon as given the chance, because it follows all available paths.
Thus an isometry does not occur for this case.

Whenever $k>3$, it appears that the tensor network does result in an isometry.
In Ref.~\cite{jahnCentralChargesAperiodic2020}, these cases were studied separately according to $n=3$ and $n>3$.
A central perfect tensor is placed to begin with, and then consecutive layers of polygons are added.
This is done such that each vertex is on the boundary of the tiling when it is in the most-recently added layer, and it is in the interior afterwards.

Then the vertices of the network can be characterized according to how many other vertices they are connected to from the same or previous layers.
Each such vertex is connected with exactly two from the same layer as itself, as it is on the boundary of the tiling.
Either (a) it connect to no vertices from previous layers (only two of those from the same layer as itself); (b) they connect to one vertex from a previous layer (and thus a total of three including those of the same layer); or (c) they can connect to two vertices of previous layers (for a total of four vertices).
We label these types of vertices as $a$, $b$ and $c$ respectively, and note that the third type $c$ only appears when $n=3$.
The initial vertex effectively counts as a distinct category, $i$, which connects to no vertices from the same or previous layers.

The construction of each new layer can be surmised using replacement rules of these vertex types as follows~\cite{jahnCentralChargesAperiodic2020}
\begin{align}
		n&=3: \begin{cases}
			a \to c b^{k-4} \\
			b \to c b^{k-5} \\
			c \to c b^{k-6} \\
		\end{cases} \\
		n&>3: \begin{cases}
			a\to a^{n-4} b\left( a^{n-3}b \right)^{k-3}\\
			b\to a^{n-4} b\left( a^{n-3}b \right)^{k-4}\\
		\end{cases}.
	\label{eq:Replacement}
\end{align}
The resultants of these replacement rules are written in a specific order to denote a particular ordering of each layer; e.g.\ one can use a clockwise convention such that the new vertices are in the appropriate order from left to right.
We also require the rules for $i$-type vertices which were not present in Ref.~\cite{jahnCentralChargesAperiodic2020} due to the different convention, which are as follows for all $n$:
\begin{align}
			i \to (b a^{n-3})^k
	\label{eq:iReplacement}
\end{align}

Our proof strategy is as follows -- we construct the network such that the initial site is the one onto which a local bulk operator is applied. This operator is then pushed through to three different choices of region $S_1$, $S_2$ and $S_3$.
Each of these regions can fully reconstruct an arbitrary bulk logical operator on the initial site, so their complements are correctable, and if the union of these complements covers the full boundary space we can apply \cref{lem:PY} to obtain the desired result.
Thus we must show that $S_1$, $S_2$ and $S_3$ can be constructed such that no site of the boundary lies in the intersection of all three regions.
We proceed by partitioning the boundary into $k=2N-1$ sets, $\left\{ M_j \right\}_{j=1}^{j=2N-1}$; each of which corresponds to some $b_j$, one of the $b$-type vertices that directly connects to the initial vertex.
Each of these vertices comes with $n-3$ $a$-type vertices from the $i$ replacement rule \cref{eq:iReplacement}.
Together, these vertices, along with all of the vertices obtained from them via repeated application of the replacement rules in \cref{eq:Replacement}, form the set $M_j$.

One can push out the bulk operator acting on the initial vertex to any choice of $N$ of the $b_j$ vertices.
Consider what would happen if one could construct a pushing protocol such that an operator acting on each $b_j$ can always be pushed out into just the boundary region $M_j$.
Then the initial bulk operator can be pushed to any $N$ of the $M_j$, and we could use these sets to form $S_1$, $S_2$ and $S_3$.
For example, let $S_1 = \bigcup_{j=1}^{j=N}$, $S_2 = \bigcup_{j=2}^{j=N+1} $, and $S_3 = \bigcup_{j=N+1}^{j=2N-1}$. 
Each $S_i$ consists of the union of $N$ of the $M_j$ regions, and no $M_j$ belongs to all three, so these choices of $S_i$ are appropriate choices that allow application of the main lemma \cref{lem:PY}, and thus no non-Cliffords can be applied transversally.

\begin{definition}[Replacement Pushing Assumption]
	An operator acting on a vertex $b_j$ connected to the initial vertex can be pushed out to $M_j$, which is defined as the image of $b_j$ under repeated applications of the replacement rules \cref{eq:Replacement}.
\end{definition}
In other words, the replacement pushing assumption means that we can partition the boundary into $2N-1$ regions $\left\{ R_j \right\}$ such that each $b_j\in \mathcal{E}[R_j]$, with the entanglement wedge $\mathcal{E}$ as defined in \cref{def:ew}.

When does the replacement pushing assumption apply?
The replacement rules are constructed such that whenever they are applied to an existing vertex and a new $b$ or $c$ vertex is obtained, it must connect to that existing vertex.
Thus, if at least $N$ vertices of $b$ or $c$ type are obtained from any replacement rule, then the replacement pushing assumption applies straightforwardly; each vertex can be pushed into its replacement, which is repeated until the boundary is reached.

First, let us consider the $n>3$ case.
The rules for $a$ and $b$ result in $k-2$ and $k-3$ new $b$-type vertices respectively.
Thus, so long as $k-3 \geq N$, i.e.\ $k \geq 7$, replacement pushing applies, and the result follows.

Now, for $n=3$ -- note that because none of the replacement rules in \cref{eq:iReplacement,eq:Replacement} produce any $a$ vertices, we only need to consider $b$ and $c$ (again, this is different to Ref.~\cite{jahnCentralChargesAperiodic2020}, where their conventions necessitate starting with some $a$ vertices).
Thus, if $N \geq k-5$, i.e.\ $k \geq 11$, replacement pushing applies, and the result follows.

The cases of $\left\{ n,k \right\}$ which produce valid HaPPY networks (i.e.\ have odd $k=2N-1$, satisfy \cref{eq:hyperbolic}, and form actual isometries i.e.\ $k\neq 3$), but for which replacement pushing does not straightforwardly apply according to the above arguments, are (A) $ \left\{ n>3, 5 \right\}$, (B) $\left\{ 3,9 \right\}$, and (C) $\left\{ 3,7 \right\}$.
In each of these cases, replacement pushing does in fact still apply, but with a little more effort required to complete the argument.

(A) In this case, replacement pushing does straightforwardly apply for $a$-type vertices, as their replacement rule gives three $b$-type vertices, and $N= \frac{k+1}{2} = 3$.
However, $b$-type vertices are potential causes for concern, as operators acting here must be pushed to $N=3$ connecting vertices, but there are only two available connected $b$-type vertices in the succeeding layer.
Thus, one leg must be pushed into \emph{along} the boundary, i.e.\ there must be some pushing to the same layer.
However, this is fine; we just push in the direction of the nearest $a$-type vertex that belongs to the same vertex set $M_j$ -- noting that because $n>3$, there always exists at least one such $a$-type vertex.
Even if the $b$-type vertex in question does not connect to an $a$-type vertex directly, it can push to a neighouring $b$-type vertex, which will eventually connect to an appropriate $a$-type from the same set $M_j$.

(B) 
Each $b$-type vertex is mapped under replacement to a $c b b b b$ string of vertices, so these straightforwardly obey the replacement pushing assumption.
Each $c$-type vertex is mapped under replacement with $c b b b$, but we must push to $N=5$ vertices.
We remark that due to the replacement rules, no two $c$ vertices are ever adjacent, so there must be a $b$ vertex to the right of each $c$ vertex.
Furthermore, this $b$ vertex must belong to the same set $M_j$, again because of the nature of the replacement rules.
Thus we simply push from each $c$-type vertex onto its four descendant vertices, and for the remaining leg, push into the $b$ vertex to its right -- a protocol which satisfies the replacement pushing assumption.

(C)
This case is slightly more complicated as we must employ different pushing rules for $b$-type vertices depending on their position.
We remark that since the $n=3$ replacement rules always generate a $c$ for each vertex, and $c$ vertices by definition connect with \emph{two} vertices from the previous layer, the $c$ in a given vertex's replacement is also connected to that of the vertex to its left.
Thus, a vertex has three options for pushing -- it can push ``forwards'' to the vertices in its replacement; it can push ``sideways'' to adjacent vertices of the same layer, or it can push ``diagonally'' to the $c$-type vertex in the replacement of its neighbor on the right.
Each $b$-type vertex is mapped under replacement to a $c b b$ string of vertices, but must be pushed along $N=4$ legs.
For any $b$-type vertex that is not on the very far right of the corresponding set $M_j$, it can be pushed into the three vertices in its replacement, and also to the $c$-type vertex in the replacement of its neighbor on the right.
For the rightmost vertex, which is always a $b$-type, it can be pushed to its left neighbour (which is always a $b$), as well as to the three in its replacement.
The rightmost vertex must always have a $b$ neighbour because of the replacement rules; although some $b$ vertices have $c$-types on either side, these are never on the boundary of the $M_j$ set.

Finally, each $c$-type vertex is mapped under replacement to $c b$, but must also be pushed along four legs.
Again, we can push into the two in its replacement, as well as the $c$-type in the replacement of its rightward neighbor.
For the final leg, we can push sideways to an adjacent $b$-vertex from the same $M_j$.

None of these pushing rules are in conflict with one another, so we have shown a consistent pushing protocol which keeps each operator in the same $M_j$; i.e.\ the replacement pushing assumption is satisfied.

We remark that all of these ideas still apply in the case of perfect planar/block perfect tensors as well, as we never need to push to non-adjacent legs.
Similar arguments can be generalized to non-homogeneous codes (such as \cref{fig:alternating}), as well as higher-dimensional HaPPY codes.

\end{document}